\documentclass[reqno,12pt]{amsart}

\newtheorem{theorem}{Theorem}[section]
\newtheorem{theorem*}[theorem]{Theorem}
\newtheorem{lemma}[theorem]{Lemma}
\newtheorem{definition}[theorem]{Definition}
\newtheorem{proposition}[theorem]{Proposition}

\theoremstyle{remark}
\newtheorem{remark}[theorem]{Remark}
\numberwithin{equation}{section}

\usepackage{array}
\usepackage{multirow}
\usepackage{array,booktabs}
\usepackage{todonotes}
\usepackage[margin=1in]{geometry}
\usepackage{verbatim}
\usepackage{amsmath}
\usepackage{amsthm}
\usepackage{amsfonts}
\usepackage{amssymb}
\usepackage{amsthm}
\usepackage{enumerate}
\usepackage{amsrefs}
\usepackage{enumerate}
\usepackage{amsrefs}
\usepackage{tensor}
\usepackage{epic}
\usepackage{mathtools}
\usepackage{tikz}
\usetikzlibrary{chains}
\usepackage{graphicx}
\usepackage[font=small,skip=0pt]{caption}
\usepackage{makecell}
\usepackage{tensor}
\usepackage[super]{nth}
\usepackage{tabulary}
\usepackage{lipsum}

\begin{document}

\title{Trigonometric $\vee$-systems and solutions of WDVV equations}

\author{ Maali  Alkadhem} 
\author{ Misha Feigin}

\address{School of Mathematics and Statistics, University of Glasgow, University Place, Glasgow G12 8QQ, UK}  

\email{m.alkadhem.1@research.gla.ac.uk; misha.feigin@glasgow.ac.uk}

%





  %

\begin{abstract}
We consider a class of trigonometric solutions of WDVV equations determined by collections of vectors with multiplicities.
We show that such solutions can be restricted to special subspaces to produce new solutions of the same type. 
We find new solutions given by restrictions of root systems, as well as examples which are not of this form. 
Further, we consider a closely related notion of a trigonometric $\vee$-system and we show that their subsystems are also trigonometric $\vee$-systems. 
Finally, while reviewing the root system case we determine a version of (generalised) Coxeter number for the exterior square of the reflection representation of a Weyl group.
\end{abstract}

\maketitle

\section{Introduction}\label{Introduction}
The Witten--Dijkgraaf--Verlinde--Verlinde (WDVV) equations are a remarkable set of nonlinear third order partial differential equations for a single function $\mathcal{F}$. They were discovered originally in two-dimensional topological field theories, and they are in the core of Frobenius manifolds theory  \cite{Dubrovin.1996 }, 
in which case prepotential $\mathcal{F}$ is a function on a Frobenius manifold $\mathcal{M}.$
A flat metric can be defined on $\mathcal{M}$ in terms of the third order derivatives of $\mathcal{F}$ which allows to reformulate WDVV equations as the associativity condition of a multiplication in a family of Frobenius algebras defined on the tangent planes $T_{\ast}\mathcal{M}.$
Structure constants of the multiplication are also given in terms of the third order derivatives of $\mathcal{F}.$

There is a remarkable class of polynomial solutions of WDVV equations, which corresponds to (finite) Coxeter groups $\mathcal{W}$ \cites{ Dubrovin.1996 }. 
In this case the space $\mathcal{M}$ is the space of $\mathcal{W}$-orbits in the reflection representation of $\mathcal{W}.$
Prepotential $\mathcal{F}$ is then a polynomial in the flat coordinates of the metric known as Saito metric. 

For any Frobenius manifold there is an almost dual Frobenius manifold   introduced by Dubrovin in \cite{ Dubrovin 2004}. Prepotentials for almost dual structures of the polynomial Frobenius manifolds can be expressed in a simple form
\begin{equation}\label{prepotential.rational.log}
\mathcal{F}=\mathcal{F}^{rat}=\sum_{\alpha\in\mathcal{A}}\alpha(x)^{2}\log \alpha(x),\quad x\in V,
\end{equation}
where $\mathcal{A}=\mathcal{R}$ is the root system of the group $\mathcal{W}$.
In this case the constant metric is 
the $\mathcal{W}$-invariant form on the vector space $V$ of the reflection representation of the group $\mathcal{W}$. 

Such solutions $\mathcal{F}$ of WDVV equations appear in four-dimensional Seiberg--Witten theory as perturbative parts of Seiberg--Witten prepotentials.    
Thus Marshakov, Mironov and Morozov found them for classical root systems in \cites{MMM.1996, MMM.2000}. 
Solutions (\ref{prepotential.rational.log}) for non-classical root systems were found by Gragert and Martini in \cite{Martini+Gragert 1999}. 

Veselov found solutions $\mathcal{F}$ of the form (\ref{prepotential.rational.log}) for some not fully symmetric configurations  of covectors $\mathcal{A}\subset V^{\ast}$, and he introduced the notion of a $\vee$-system \cite{Veselov 1999} formulated in terms of linear algebra. A configuration of vectors $\mathcal{A}$ is a $\vee$-system exactly when the corresponding prepotential (\ref{prepotential.rational.log}) satisfies WDVV equations \cites{ Veselov 1999,  Misha&Veselov 2008}.
This property can also be reformulated in terms of flatness of a connection on the tangent bundle $T_{\ast}V$ \cite{Veselov 2000}.
A closely related notion of the Dunkl system was introduced and studied in \cite{CHL 2005}.
That structure for complex reflection groups was investigated further in \cite{Arsie+ Lorenzoni 2017} in relation with Frobenius manifolds theory.

The class of $\vee$-systems is closed under the natural operations of taking subsystems \cite{Misha&Veselov 2008} and under restriction of a system to the intersection of some of the hyperplanes $\alpha(x)=0,$ where $\alpha \in \mathcal{A}$ \cite{Misha&Veselov 2007}.
The class of $\vee$-systems contains multi-parameter deformations of the root systems $A_n$ and $B_n$ (\cite{Chalykh+ Veselov 2001}, see also \cite{Misha&Veselov 2008} for more examples). 
The underlying matroids were examined in \cite{Schreiber+ Veselov 2014}. 
The problem of classification of $\vee$-systems remains open.

In this work we are interested in the trigonometric solutions of WDVV equations which have the form 
\begin{equation}\label{trig. general solution}
    \mathcal{F}=\mathcal{F}^{trig}=\sum_{\alpha \in \mathcal{A}}c_{\alpha} f( \alpha(x))+ Q,
\end{equation}
where $f^{'''}(z)=\cot{z}, c_{\alpha}\in \mathbb{C},$ and $Q=Q(x,y)$ is a cubic polynomial which depends on the additional variable $y \in \mathbb{C}$. 
Such solutions appear in five-dimensional Seiberg--Witten theory as perturbative parts of prepotentials \cite{MMM.2000}. 
Solutions of the form (\ref{trig. general solution}) for (non-reduced) root systems $\mathcal{A=\mathcal{R}}$ of Weyl groups and $\mathcal{W}$-invariant multiplicities $c_{\alpha}$ were found by  Hoevenaars and Martini in \cites{Martini 2003 (1), Martini 2003}. 
They appear as prepotentials for the almost dual Frobenius manifold structures on the extended affine Weyl groups orbit spaces \cites{ Dubrovin+ Zhang 1998,  Dubrovin+Strachan+ Zhang+Zuo 2019}-- see \cite{Riley+ Strachan 2007} for type $A_n$.  
In some cases such solutions may be related to the rational solutions (\ref{prepotential.rational.log}) by twisted Legendre transformations \cite{Riley+ Strachan 2007}.

Bryan and Gholampour found another remarkable appearance of trigonometric solutions (\ref{trig. general solution}) in geometry as they studied quantum cohomology of resolutions of $A,D,E$ singularities \cite{Bryan 2008}.
The associative quantum product on these cohomologies is governed by the corresponding solutions $\mathcal{F}^{trig}$  with $\mathcal{A}=A_n, D_n, E_n$ respectively.

Solutions of WDVV equations of the form (\ref{trig. general solution}) without full Weyl symmetry were considered by one of the authors in \cite{Misha2009} where the notion of a trigonometric $\vee$-system was introduced and its close relation with WDVV equations was established. 
A key difference with the rational case is the existence of a rigid structure of a series decomposition of vectors from $\mathcal{A}$ which generalizes the notion of strings for root systems. 

Many-parameter deformations of solutions $F^{trig}$ for the classical root systems were obtained by Pavlov from reductions of Egorov hydrodynamic chains \cite{Pavlov 2006}.
Closely related many-parameter family of flat connections in type $A_n$ was considered by Shen in \cites{Shen 2018, Shen 2019}.

Study of the trigonometric and rational cases is related since if a configuration $\mathcal{A}$ with collection of multiplicities $c_{\alpha}, \alpha \in \mathcal{A}$ is a trigonometric $\vee$-system then configuration $\sqrt{c_{\alpha}}\alpha$ is a rational one \cite{Misha2009}.
However, due to the presence of the extra variable $y$ in the trigonometric case it is already nontrivial for $\dim V=2$ while the smallest nontrivial dimension of $V$ in the rational case is $3$.  

There is also an important class of elliptic solutions of WDVV equations, which was considered by Strachan in \cite{Strachan 2010} where, in particular, certain solutions related to $A_n$ and $B_n$ root systems were found.
The prepotentials appear as almost dual prepotential associated to Frobenius manifold structures on $A_n$ and $B_n$ Jacobi groups orbit spaces \cites{Bertola 1999, Bertola paper}. 
Such solutions appear also in six-dimensional Seiberg--Witten theory \cite{BMMM 2007}.

In this paper we study trigonometric solutions $\mathcal{F}^{trig}$ of the form (\ref{trig. general solution}) of WDVV equations.
In section \ref{section.Trig.sys and WDVV} we recall the notion of a trigonometric $\vee$-system and revisit its close relation with solutions of WDVV equations. 

We investigate operations of taking subsystems and restrictions in section \ref{section.subsystems of trig.systems}, \ref{section.restriction of trig.systems}. We show that a subsystem of a trigonometric $\vee$-system is also a trigonometric $\vee$-system, and that one can restrict solutions of WDVV equations of the form (\ref{trig. general solution}) to the intersections of hyperplanes to get new solutions. 

In section  \ref{section.BCn} we find solutions $\mathcal{F}^{trig}$ for the root system $BC_n$ which depend on three parameters. By applying restrictions we obtain in sections \ref{section.BCn} and \ref{section.An} multi-parameter families of solutions $\mathcal{F}^{trig}$ for the classical root systems thus recovering and extending results from \cite{Pavlov 2006}. 
In the case of $BC_n$ we get a family of solutions depending on $n+3$ parameters which can be specialized to Pavlov's $(n+1)$-parametric family from \cite{Pavlov 2006}.
A related multi-parameter deformation of $BC_n$ solutions (\ref{trig. general solution}) when $Q$ depends on $x$ variables only was obtained in \cite{MGM 2020} by similar methods.

In section \ref{section.4-dim trig.systems} we consider solutions $\mathcal{F}^{trig}$ for $n \leq 4.$
We show that solutions with up to five vectors on the plane belong to deformations of classical root systems.
We also get new examples of solutions $\mathcal{F}^{trig}$ of the form (\ref{trig. general solution}) some of which cannot be obtained as restrictions of solutions (\ref{trig. general solution}) for the root systems. 

In section~\ref{section.root systems solutions revisited} we revisit solutions $\mathcal{F}^{trig}$ for the root systems studied in \cites{Martini 2003 (1) , Martini 2003, Bryan 2008, Shen 2018, Shen 2019}.
The polynomial $Q$ in this case depends on a scalar $\gamma_{(\mathcal{R},c)}$ which is determined in these papers for any invariant multiplicity function $c\colon \mathcal{R}\to \mathbb{C}.$
We give a formula for $\gamma_{(\mathcal{R},c)}$ in terms of the highest root of $\mathcal{R}$ generalizing a statement from \cite{Bryan 2008} for special multiplicities. 
We also find a related scalar $\lambda_{(\mathcal{R},c)}$ which is invariant under linear transformations applied to the root system $\mathcal{R}.$ 
This scalar may be thought of as a version of generalized Coxeter number (see e.g. \cite{Misha 2012}) for the irreducible $\mathcal{W}$-module $\Lambda^{2}V$ since it is given as a ratio of two canonical $\mathcal W$-invariant symmetric bilinear forms on $\Lambda^{2}V$. 
\section{Trigonometric $\vee$-systems and WDVV equations}\label{section.Trig.sys and WDVV}

Let $V$ be a vector space of dimension $N$ over $\mathbb{C}$ and let $V^{\ast}$ be its dual space. Let $\mathcal{A}$ be a finite collection of covectors $\alpha \in V^{\ast}$ which belongs to a lattice of rank $N.$  

Let us also consider a multiplicity function $c \colon \mathcal{A} \to \mathbb{C}.$ We denote $c(\alpha)$ as $c_{\alpha}.$ 
We will assume throughout that the corresponding symmetric bilinear form
%
 $$G_{(\mathcal{A},c)}(u,v)\coloneqq \sum_{\alpha\in\mathcal{A}}c_{\alpha}\alpha(u)\alpha(v),\quad u,v\in V$$
%
is non-degenerate. We will also write $ G_{\mathcal{A}}$ for $ G_{(\mathcal{A},c)}$ to simplify notations. 
The form $ G_{\mathcal{A}}$ establishes an isomorphism 
$\phi \colon V\to V^{\ast},$ and we denote the inverse $\phi^{-1}(\alpha)$ by $\alpha^{\vee},$ where $ G_{\mathcal{A}}(\alpha^{\vee},v)=\alpha (v)$ for any $v \in V.$

Let $U\cong \mathbb{C}$ be a one-dimensional vector space.
We choose a basis in $V \oplus U$ such that $e_1,\dots,e_N$ is a basis in $V$ and $e_{N+1}$ is the basis vector in $U,$
and let $x_1, \dots,  x_{N+1}$ be the corresponding coordinates.
We represent vectors $x\in V, y\in U$ as $x=(x_{1},...,x_{N})$ and $y=x_{N+1}.$ 
Consider a function $F \colon V \oplus U \to \mathbb{C}$ of the form

\begin{equation}\label{F with extra variable y}
F=\frac{1}{3}y^3+\sum_{\alpha\in\mathcal{A}}c_{\alpha}\alpha(x)^{2}y+\lambda\sum_{\alpha\in\mathcal{A}}c_{\alpha}f(\alpha(x)),
\end{equation}
where $\lambda \in \mathbb{C}^{\ast}$ and function $f(z)= \frac{1}{6} i z^3+\frac{1}{4} Li_{3}(e^{-2iz})$ satisfies $f^{\prime \prime \prime}(z)=\cot z$. 
The WDVV equations is the following system of partial differential equations

\begin{equation}\label{WDVV with y}
\ F_{i}F_{N+1}^{-1}F_{j}=F_{j}F_{N+1}^{-1}F_{i},\quad i,j=1,...,N,
\end{equation}
where $F_{i}$ is $(N+1)\times(N+1)$ matrix with entries
$(F_i)_{p q}=\frac{\partial^{3}F}{\partial x_{i}\partial
x_{p}\partial x_{q}}$ ($p,q =1,\dots, N+1$). 

Let  ${e^{1},...,e^{N}}$ be the basis in $V^{\ast}$ dual to the basis ${e_{1},...,e_{N}}\in V.$ 
Then for any covector $\alpha\in V^{\ast}$ we have 
$\alpha=\displaystyle\sum_{i=1}^{N}\alpha_{i}e^{i}$ 
and $\alpha^{\vee}=\displaystyle\sum_{i=1}^{N} \alpha^{\vee}_{i}e_{i},$
where $\alpha_{i},\alpha_{i}^{\vee}\in \mathbb{C}.$ 
Then
\begin{equation}\label{The matrix of N+1}
F_{N+1}=2\begin{pmatrix}
\displaystyle\sum_{\alpha\in\mathcal{A}} c_{\alpha}\alpha \otimes\alpha &0 \\
0 &1 
\end{pmatrix},    
\end{equation}
where we denoted by $\alpha$ both column and row vectors $\alpha =(\alpha_{1},...,\alpha_{N}),$ and $\alpha\otimes\alpha$ is $N \times N$ matrix with matrix entries $(\alpha\otimes\alpha)_{jk}=\alpha_{j} \alpha_{k}.$
Let us define
\begin{equation}\label{eta}
   \eta_{ij} = (F_{N+1})_{ij}, \quad \eta^{ij} =(F_{N+1}^{-1})_{ij},
\end{equation}
where $i,j=1,\dots, N+1.$
Now we will establish a few lemmas which will be useful later. The next statement is standard.
\begin{lemma}\label{alpha and its dual}
Let $\widetilde{G}$ be the matrix of the bilinear form $G_{\mathcal{A}}$, that is its matrix entry 
$(\widetilde{G})_{ij}=G_{\mathcal{A}}(e_{i},e_{j}),$ where $ i,j=1,\dots,N.$
Then for any covector $\gamma =(\gamma_{1},...,\gamma_{N})\in V^{\ast}$ and $\gamma^{\vee} =(\gamma_{1}^{\vee},...,\gamma_{N}^{\vee})\in V,$ we have
$\widetilde{G}^{-1}\gamma^{T}=(\gamma^{\vee})^T.$
\end{lemma}

Let $M_{\mathcal{A}}=V \setminus \cup_{\alpha \in \mathcal{A}}\Pi_{\alpha}$ be the complement to the union of all the hyperplanes $\Pi_{\alpha}\coloneqq \{ x\in V\colon \alpha(x)=0 \}.$
For any vector $\overline{a}=(a_{1},\dots,a_{N+1})\in V \oplus U$ let us introduce the corresponding vector field 
$\partial_{\overline{a}} = \sum_{i=1}^{N+1}a_{i}\partial_{x_{i}}\in T_{\ast}( V \oplus U).$
For any $\overline{b}=(b_{1},\dots,b_{N+1})\in V \oplus U$
we define the following multiplication on the tangent space $T_{(x,y)}( M_{\mathcal{A}} \oplus U)$:
\begin{equation} \label{algebra in V+U}
\partial_{\overline{a}} \ast \partial_{\overline{b}}=a_{i}b_{j}\eta^{kl}F_{ijk} \partial_{x_{l}},\quad i,j,k,l=1,...,N+1,
\end{equation}
where $\eta^{kl}$ is defined in (\ref{eta}) and the summation over repeated indices here and below is assumed. It is clear from the definition that the multiplication $\ast$ is commutative and distributive.
The proof of the next statement is standard (see \cite{Dubrovin.1996} for a similar statement).
\begin{lemma}\label{associyivity and WDVV}
The associativity of multiplication $\ast$ is equivalent to the WDVV equation (\ref{WDVV with y}).
\end{lemma}
Let us introduce vector field $E$ by 
\begin{equation*}\label{E vector field}
    E=\partial_{x_{N+1}}\in T_{\ast}(V\oplus U).
\end{equation*}
\begin{proposition}\label{identity of algebra}
Vector field $E$ is the identity for the multiplication (\ref{algebra in V+U}).
\end{proposition}
\begin{proof}

For all $1\leq i \leq N+1$ we have 
$$\partial_{x_{i}} \ast E =\eta^{kj}F_{i,N+1,j}\partial_{x_{k}}=\eta^{kj}\eta_{ij} \partial_{x_{k}}\\
=\partial_{x_{i}}.$$
\end{proof}

\begin{proposition}\label{explicit formula of the product 1}
Let $a=(a_1, \dots, a_N),b=(b_1, \dots, b_N)\in V,$ and let 
$\partial_{a}=\sum_{i=1}^{N}a_{i}\partial_{x_{i}},\quad$
${\partial_{b}=\sum_{i=1}^{N}b_{i}\partial_{x_{i}}}.$
Then the product (\ref{algebra in V+U}) has the following explicit form
\begin{equation}\label{explicit formula 1.1}
\partial_{a}\ast \partial_{b}= \sum_{\alpha \in \mathcal{A}}c_{\alpha} \alpha(a)\alpha(b)(\frac{\lambda}{2}  \cot \alpha(x)\partial_{\alpha^{\vee}}+E). 
\end{equation}
\end{proposition}
\begin{proof}
Note that $\eta^{m,N+1}=\frac{1}{2}\delta_{m}^{N+1}$ for any $m=1,...,N+1,$ where $\delta_{i}^{j}$ is the Kronecker symbol. Therefore from (\ref{algebra in V+U}) we have
\begin{align*} 
\partial_{a} \ast \partial_{b} &=a_{i}b_{j}(\sum_{k,l=1}^{N}\eta^{kl}F_{ijk}\partial_{x_{l}}+\frac{1}{2} F_{i,j,N+1}\partial_{x_{N+1}}),
\end{align*}
where
$$F_{ijk}=\lambda \sum_{\alpha \in \mathcal{A}}c_{\alpha} \alpha_{i}\alpha_{j}\alpha_{k}\cot \alpha(x).$$ 
Then we have
\begin{align} \label{53}
&  \sum_{k,l=1}^{N} a_{i}b_{j}\eta^{kl} F_{ijk}\partial_{x_{l}} =\lambda \sum_{\alpha \in \mathcal{A}}\sum_{k,l=1}^{N}c_{\alpha}\alpha(a)\alpha(b)\alpha_{k}\eta^{kl} \cot \alpha(x)\partial_{x_{l}}=\frac{\lambda}{2} \sum_{\alpha \in \mathcal{A}}c_{\alpha} \alpha(a)\alpha(b)\cot \alpha(x)\partial_{\alpha^{\vee}}
\end{align}
by Lemma \ref{alpha and its dual}. 
Also by formula (\ref{The matrix of N+1}) we have that
\begin{align} \label{56}
\frac{1}{2} \sum_{i,j=1}^{N} a_{i}b_{j} F_{i,j,N+1} &=\sum_{\alpha \in \mathcal{A}}c_{\alpha} \alpha(a)\alpha(b).
\end{align}
The statement follows from formulas (\ref{53}) and (\ref{56}).
\end{proof}
If we identify vector space $ V \oplus U$ with the tangent space $T_{(x,y)}( V \oplus U)\cong V \oplus U$, then multiplication (\ref{explicit formula 1.1}) can also be written as
\begin{equation}\label{a*b in V+U}
a\ast b= \sum_{\alpha \in \mathcal{A}}c_{\alpha} \alpha(a)\alpha(b)(\frac{\lambda}{2} \cot \alpha(x)\alpha^{\vee}+E). 
\end{equation}
Now for each vector $\alpha \in \mathcal{A}$ let us introduce the set of its collinear vectors from $\mathcal{A}$:
%
$$    \delta_{\alpha} \coloneqq \{\gamma \in \mathcal{A} \colon \gamma 	\sim \alpha   \}.$$
%
Let $\delta \subset \delta_{\alpha}$ and $\alpha_{0}\in \delta_{\alpha}.$
Then for any $\gamma \in \delta$ we have $\gamma= k_{\gamma}\alpha_{0}$ for some $k_{\gamma}\in \mathbb{R}.$ 
Note that $k_{\gamma}$ depends on the choice of $\alpha_{0}$ and different choices of $\alpha_{0}$ give rescaled collections of these parameters.
Define $C_{\delta}^{\alpha_{0}}\coloneqq \displaystyle \sum_{\gamma\in \delta}c_{\gamma}k_{\gamma}^{2}.$ 
Note that  $C_{\delta}^{\alpha_{0}}$ is non-zero if and only if  $C_{\delta}^{\widetilde{\alpha}_{0}}\neq 0$ for any $\widetilde{\alpha}_{0} \in \delta.$

For any $\alpha,\beta \in \mathcal{A}$ we define
$B_{\alpha,\beta}\coloneqq \alpha \wedge \beta =\alpha \otimes \beta - \beta \otimes \alpha \in \Lambda^{2}V^{\ast}, $
then
$B_{\alpha, \beta}\colon V \otimes V \to \mathbb{C}$ such that
$B_{\alpha, \beta}(a\otimes b)=\alpha\wedge\beta(a\otimes b)=\alpha(a)\beta(b)-\alpha(b)\beta(a)$ for any $a,b \in V.$
The following proposition holds.
\begin{proposition}\label{identity for A with collinear vectors}
Assume that prepotential (\ref{F with extra variable y}) satisfies the WDVV equations (\ref{WDVV with y}). 
Suppose that  $C_{\delta_{\alpha}}^{\alpha_{0}}\neq 0$ for any $\alpha \in \mathcal{A}, \alpha_{0}\in \delta_{\alpha}.$ 
Then the identity 
\begin{equation} \label{identity for A.1}
\sum_{\beta \in \mathcal{A}\setminus \delta_{\alpha} }c_{\beta} \alpha(\beta^{\vee})\cot \beta(x)B_{\alpha , \beta}(a \otimes b)\alpha \wedge \beta =0
\end{equation}
holds for all $a,b \in V$ provided that $\alpha(x)=0.$
\end{proposition}
\begin{proof}
For any $a=(a_{1},...,a_{N}) \in V$ we define $F_{a}=\displaystyle\sum_{i=1}^{N}a_{i}F_{i}.$ 
Also we define the matrix ${F_{a}^{\vee}=F_{N+1}^{-1}F_{a}}.$ 
The WDVV equations (\ref{WDVV with y}) are equivalent to the commutativity $[F_{a}^{\vee},F_{b}^{\vee}]=0$ for any $a,b \in V.$
The commutativity $[F_{a}^{\vee},F_{b}^{\vee}]=0$ is then equivalent to the identities \cite{Misha2009}
\begin{align} 
\nonumber 
\sum_{\gamma ,\beta \in\mathcal{A}}c_{\gamma} c_{\beta} G_{\mathcal{A}}(\gamma^{\vee},\beta^{\vee})B_{\gamma,\beta}(a \otimes b)\cot\gamma(x)\gamma^{\vee} &=0, \\
\label{second iden Misha}
\sum_{\gamma ,\beta \in\mathcal{A}}\Big(\frac{\lambda^{2}}{4} c_{\gamma} c_{\beta} G_{\mathcal{A}}(\gamma^{\vee},\beta^{\vee})\cot \gamma(x)\cot \beta(x)+c_{\gamma}c_{\beta}\Big)B_{\gamma,\beta}(a \otimes b)\gamma \wedge \beta &=0.
\end{align}
Let us consider terms in the left-hand side of the relation (\ref{second iden Misha}), where $\beta$ or $\gamma$ is proportional to $\alpha.$ The sum of these terms has to be regular at $\alpha(x)=0.$
This implies that the product 
\begin{equation}\label{product of PH}
 \Big( {\lambda^{2}}\displaystyle \sum_{\gamma \in \delta_{\alpha}} k_{\gamma}^{3}c_{\gamma} \cot \gamma(x)  \Big) \Big( \sum_{\beta \in \mathcal{A} \setminus \delta_{\alpha}}c_{\beta} \alpha_{0}(\beta^{\vee}) \cot \beta(x) B_{\alpha_{0},\beta}(a\otimes b)\alpha_{0} \wedge \beta  \Big)
\end{equation}
is regular at $\alpha(x)=0.$
The first factor in the product (\ref{product of PH}) has the first order pole at $\alpha(x)=0$ by the assumption that  $C_{\delta_{\alpha}}^{\alpha_{0}}\neq 0$ for any $\alpha \in \mathcal{A}, \alpha_{0}\in \delta_{\alpha}.$ This implies the statement.
\end{proof}
Similarly to Proposition \ref{identity for A with collinear vectors} the following proposition can also be established.
\begin{proposition}\label{New version.identity for A with collinear vectors}
Assume that prepotential (\ref{F with extra variable y}) satisfies the WDVV equations (\ref{WDVV with y}). 
Suppose that  $C_{\delta}^{\alpha_{0}}\neq 0$ for any $\alpha \in \mathcal{A},\delta\subset \delta_{\alpha}, \alpha_{0}\in \delta_{\alpha}.$ 
Then the identity (\ref{identity for A.1}) holds for any $a,b \in V$ provided that $\tan \alpha(x)=0.$
\end{proposition}
The proof is similar to the proof of Proposition \ref{identity for A with collinear vectors}. Indeed, we have that expression (\ref{product of PH}) is regular at $\alpha(x)=\pi m, m\in \mathbb{Z}.$ Assumptions imply that the first factor in (\ref{product of PH}) has the first order pole, which implies the statement.

The WDVV equations for a function $F$ can be reformulated using geometry of the configuration $\mathcal{A}.$ Such a geometric structure is embedded in the notion of a trigonometric $\vee$-system.
Before defining trigonometric $\vee$-system precisely we need a notion of \textit{series} (or {\it strings}) of vectors (see \cite{Misha2009}). 

For any $\alpha \in \mathcal{A}$ let us distribute all the covectors in $\mathcal{A}\setminus \delta_{\alpha}$ into a disjoint union of $\alpha$-series
$$\mathcal{A}\setminus \delta_{\alpha}= \bigsqcup_{s=1}^{k}  \Gamma_{\alpha}^{s}$$
where $k \in \mathbb{N}$ depends on $\alpha.$ These series $\Gamma_{\alpha}^{s}$ are determined by the property that for any $s=1,\dots, k$ and for any two covectors $\gamma_{1},\gamma_{2}\in \Gamma_{\alpha}^{s}$ one has either $\gamma_{1}+\gamma_{2}=m\alpha$ or $\gamma_{1}-\gamma_{2}=m\alpha$ for some $m \in \mathbb{Z}.$
We assume that the series are maximal, that is if $\gamma \in \Gamma_{\alpha}^{s}$ for some $s\in \mathbb{N},$ then $\Gamma_{\alpha}^{s}$ must contain all the covectors of the form $\pm \gamma+m\alpha\in \mathcal{A}$ with $m\in \mathbb{Z}.$ Note that if for some $\beta \in \mathcal{A}$ there is no $\gamma \in \mathcal{A}$ such that $\beta \pm \gamma=m\alpha$ for $m\in \mathbb{Z},$ then $\beta$ itself forms a single $\alpha$-series.

By replacing some vectors from $\mathcal{A}$ with their opposite ones and keeping the multiplicity unchanged one can get a new configuration whose vectors belong to a half-space. We will denote such a system by $\mathcal{A}_{+}.$ %
If this system contains repeated vectors $\alpha$ with multiplicities $c_{\alpha}^{i}$ then we replace them with the single vector $\alpha$ with multiplicity $c_{\alpha}\coloneqq \sum_{i} c_{\alpha}^{i}.$
\begin{definition}\cite{Misha2009}
The pair $(\mathcal{A},c)$ is called a trigonometric $\vee$-system if for all $\alpha \in \mathcal{A}$ and for any $\alpha$-series $\Gamma_{\alpha}^{s},$ one has the relation
\begin{equation}\label{Trig.condition}
\sum_{\beta\in \Gamma_{\alpha}^{s} }c_{\beta}\alpha(\beta^{\vee})\alpha \wedge \beta=0.
\end{equation}
\end{definition}

Note that if $\beta_{1},\beta_{2}\in \Gamma_{\alpha}^{s}$ for some $\alpha,s,$ then $\alpha\wedge\beta_{1}=\pm\alpha\wedge\beta_{2}$ so the identity (\ref{Trig.condition}) may be simplified by cancelling wedge products. 
We also note that if $\mathcal{A}$ is a trigonometric $\vee$-system then  $\mathcal{A}_{+}$ is the one as well. 

The close relation between the notion of a trigonometric $\vee$-system and solutions of WDVV equations is explained by the next theorem. Before we formulate it let us introduce two symmetric bilinear forms $G_{\mathcal{A}}^{(i)}=G_{(\mathcal{A},c)}^{(i)}, i=1,2,$ on the vector space $\Lambda^{2}V\subset V \otimes V.$

Let us consider the bilinear form $G_{\mathcal{A}}^{(1)}$ on $\Lambda^{2}V$ given by 
\begin{equation}\label{G1 version.2}
G_{\mathcal{A}}^{(1)}(z,w)=\sum_{\alpha, \beta \in \mathcal{A}}c_{\alpha}c_{\beta}B_{\alpha,\beta}(z)B_{\alpha,\beta}(w),  
\end{equation}
where $z,w \in \Lambda^{2}V.$ It is easy to see that for $z=u_{1} \wedge v_{1}, w= u_{2} \wedge v_{2},$ where $u_{1},u_{2},v_{1},v_{2}\in V$ we have
%
 $$G_{\mathcal{A}}^{(1)}(z,w)=8 \Big(G_{\mathcal{A}}(u_{1},u_{2}) G_{\mathcal{A}}(v_{1},v_{2})-G_{\mathcal{A}}(u_{1},v_{2}) G_{\mathcal{A}}(u_{2},v_{1})  \Big),$$
%
which is a natural extension of the bilinear form $G_{\mathcal{A}}$ to the space $\Lambda^{2}V.$
It is also easy to see that this form $G_{\mathcal{A}}^{(1)}$ is non-degenerate and that it is $\mathcal{W}$-invariant.
Let us also define the following bilinear form $G_{\mathcal{A}_{+}}^{(2)}$ on $ \Lambda^{2}V$:
\begin{equation}\label{second extended form}
G_{\mathcal{A}_{+}}^{(2)}\big(z,w)=\sum_{\alpha, \beta \in \mathcal{A}_{+}}c_{\alpha}c_{\beta} G_{\mathcal{A}}(\alpha^{\vee},\beta^{\vee}) B_{\alpha,\beta}(z)B_{\alpha,\beta}(w),
\end{equation}
where $z,w \in \Lambda^{2}V.$

The following statement shows that the bilinear form $G_{\mathcal{A}_{+}}^{(2)}$ is independent of the choice of the positive system $\mathcal{A}_{+}.$
\begin{lemma}\label{G2 with two positive systems}
For any positive systems $\mathcal{A}_{+}^{(1)},\mathcal{A}_{+}^{(2)}$ for a trigonometric $\vee$-system $(\mathcal{A},c)$ we have $G_{{\mathcal{A}}_{+}^{(1)}}^{(2)}= G_{{\mathcal{A}}_{+}^{(2)}}^{(2)}.$
\end{lemma}
\begin{proof}
Suppose firstly that two positive systems $\mathcal{A}_{+}^{(1)},\mathcal{A}_{+}^{(2)}$ for a trigonometric $\vee$-system $(\mathcal{A},c)$ satisfy the condition
%
$$\mathcal{A}_{+}^{(2)}=\Big( \mathcal{A}_{+}^{(1)}\setminus \delta_{\alpha} \Big) \cup \Big(-\delta_{ \alpha}\Big)$$ 
for some $\alpha\in \mathcal{A}_{+}^{(1)} .$
Notice that vector $\alpha$ cannot be a linear combination of vectors in $\mathcal{A}_{+}^{(1)}\setminus \delta_{\alpha}.$
Hence for each $\alpha$-series $\Gamma_{\alpha}^{s}$ in $\mathcal{A}_{+}^{(1)}$ we have
\begin{equation}\label{trig.condition on A1}
    \sum_{\beta \in \Gamma_{\alpha}^{s}} c_{\beta}\alpha(\beta^{\vee})=0
\end{equation}
since $B_{\alpha, \beta_{1}}=B_{\alpha, \beta_{2}}$ for all $\beta_{1},\beta_{2}\in \Gamma_{\alpha}^{s}.$

Let us consider terms in $G_{{\mathcal{A}}_{+}^{(1)}}^{(2)}(z,w)$ which contain $\alpha.$ 
They are proportional to
$$\sum_{\beta \in \mathcal{A}_{+}^{(1)}} c_{\beta}G_{\mathcal{A}}(\alpha^{\vee},\beta^{\vee})B_{\alpha,\beta}(z)B_{\alpha,\beta}(w)=\sum_{s} \sum_{\beta \in \Gamma_{\alpha}^{s}} c_{\beta}\alpha(\beta^{\vee})B_{\alpha,\beta}(z)B_{\alpha,\beta}(w)=0$$
by (\ref{trig.condition on A1}).
The statement follows in this case. 

In general, the system $\mathcal{A}_{+}^{(2)}$ can be obtained from  the system $\mathcal{A}_{+}^{(1)}$ by a sequence of steps where in each one we replace the subset of vectors $\delta_{\alpha}$ with vectors $-\delta_{\alpha}$ and the resulting system is still a positive one. In order to see this one moves continuously the hyperplane defining $\mathcal{A}_{+}^{(1)}$ into the hyperplane $\mathcal{A}_{+}^{(2)}$ so that at each moment the hyperplane contains at most one vector from $\mathcal{A}$ up to proportionality. The statement follows from the case considered above.  
\end{proof}
As a consequence of Lemma \ref{G2 with two positive systems} we can and will denote the form $G_{{\mathcal{A}}_{+}}^{(2)}$ as $G_{\mathcal{A}}^{(2)}.$ 

A close relation between trigonometric $\vee$-systems and solutions of WDVV equations is given by the following theorem.

\begin{theorem}\label{correction of Misha's Theorem} \textup{(cf.\cite{Misha2009})}
Suppose that a configuration $(\mathcal{A},c)$ satisfies the condition $C_{\delta}^{\alpha_{0}}\neq 0$ for all $\alpha \in \mathcal{A}, \, \delta\subset \delta_{\alpha},\, \alpha_{0}\in \delta_{\alpha}.$ 
Then WDVV equations (\ref{WDVV with y}) for the function (\ref{F with extra variable y}) imply the following two conditions:
\begin{enumerate}
\item $\mathcal{A}$ is a trigonometric $\vee$-system,
\item
Bilinear forms \eqref{G1 version.2}, \eqref{second extended form}  satisfy proportionality $ G_{\mathcal{A}}^{(1)}= \frac{\lambda^2}{4} G_{\mathcal{A}}^{(2)}.
$

%
%
%
\end{enumerate}
Conversely, if a configuration $(\mathcal{A},c)$ satisfies conditions (1) and (2) then WDVV equations  (\ref{WDVV with y}) hold. 
\end{theorem}
The key part of the proof is to derive trigonometric conditions from WDVV equations, which goes along the following lines (see \cite{Misha2009} for details). 
By Proposition  \ref{New version.identity for A with collinear vectors} identity (\ref{identity for A.1}) holds if $\tan \alpha(x)=0$. 
The identity (\ref{identity for A.1}) is a linear combination of $\cot \beta(x)|_{\tan \alpha(x)=0},$ which can vanish only if it vanishes for each $\alpha$-series. Hence identity (\ref{identity for A.1}) implies relations (\ref{Trig.condition}) so $\mathcal{A}$ is a trigonometric $\vee$-system.
\begin{remark}
A version of Theorem \ref{correction of Misha's Theorem} is given in \cite{Misha2009}*{Theorem 1} without specifying conditions $C_{\delta}^{\alpha_{0}}\neq 0.$ However these assumptions seem needed in general in order to derive trigonometric $\vee$-conditions for $\alpha$-series in the case when $\delta_{\alpha}\setminus \{\pm \alpha \}\neq \emptyset $ as above arguments and proofs of Propositions \ref{identity for A with collinear vectors}, \ref{New version.identity for A with collinear vectors} explain. 
\end{remark}
An important class of solutions of WDVV equations is given by (crystallographic) root systems $\mathcal{A}=\mathcal{R}$ of Weyl groups $\mathcal{W}.$ Recall that a root system $\mathcal{R}$ satisfies the property 
\begin{equation}\label{orthogonal reflection}
    s_{\alpha}\beta=\beta - \frac{2 \langle \alpha, \beta\rangle}{ \langle \alpha, \alpha \rangle}\alpha \in \mathcal{R}
\end{equation}
for any  $\alpha, \beta \in \mathcal{R},$ and one has  $\frac{2\langle \alpha, \beta\rangle}{ \langle \alpha, \alpha \rangle}\in \mathbb{Z},$
where $ \langle \cdot, \cdot \rangle$ is a $\mathcal{W}$-invariant scalar product on $V^{\ast}\cong V.$ The corresponding Weyl group is generated by reflections $s_{\alpha}, \alpha \in \mathcal{R}.$ 

The following statement was established in \cite{Martini 2003} for the non-reduced root systems.
\begin{theorem}\label{root systems solve WDVV} \textup{(cf. \cite{Martini 2003})}
Let $\mathcal{A}=\mathcal{R}$ be an irreducible root system with the Weyl group $\mathcal{W}$ and suppose that the multiplicity function $c\colon \mathcal{R} \to \mathbb{C}$ is $\mathcal{W}$-invariant. Then prepotential (\ref{F with extra variable y}) satisfies WDVV equations (\ref{WDVV with y}) for some $\lambda \in \mathbb{C}$ .
\end{theorem}
Let us explain a proof of this statement different from \cite{Martini 2003} by making use the notion of a trigonometric $\vee$-system and Theorem \ref{correction of Misha's Theorem}. 
\begin{proposition}\label{root system is trig. system}
Root system $\mathcal{A}=\mathcal{R}$ with $\mathcal{W}$-invariant multiplicity function $c$ is a trigonometric $\vee$-system.
\end{proposition}
\begin{proof}
Fix $\alpha \in \mathcal{R}.$ Take any $\beta \in \mathcal{R},$ and let $\gamma=s_{\alpha}\beta.$ Then from (\ref{orthogonal reflection}) we have that
$$\beta - \gamma =m \alpha, \quad m \in \mathbb{Z}.$$ 
Hence  $\beta, \gamma \in \Gamma_{\alpha}^{s}$ for some $s.$
%
%
The bilinear form $G_{\mathcal{R}}$ is $\mathcal{W}$-invariant so is proportional to $\langle \cdot, \cdot \rangle .$ Therefore we have 
$$c_{\beta}=c_{\gamma},\quad G_{\mathcal{R}}(\alpha, \beta)=-G_{\mathcal{R}}(\alpha,\gamma), \quad \alpha \wedge \beta = \alpha \wedge \gamma.$$
Hence,
$$c_{\beta} G_{\mathcal{R}}(\alpha, \beta) \alpha \wedge \beta + c_{\gamma} G_{\mathcal{R}}(\alpha,\gamma)\alpha \wedge \gamma =0,$$
which implies trigonometric $\vee$-conditions (\ref{Trig.condition}).
\end{proof}
It is easy to see that the bilinear form $G_{\mathcal{R}}^{(1)}$ is $\mathcal{W}$-invariant, and the same is true for the bilinear form $G_{\mathcal{R}}^{(2)}$ (see e.g. \cite{George+Misha 2019}*{Proposition 4.6}). 
Since $\mathcal{W}$-module $\Lambda^{2}V$ is irreducible, the forms $G_{\mathcal{R}}^{(1)}$ and $G_{\mathcal{R}}^{(2)}$ have to be proportional.   
By Theorem \ref{correction of Misha's Theorem} this implies Theorem \ref{root system is trig. system} provided that the form $G_{\mathcal{R}}^{(2)}$ is non-zero. 
The latter fact is claimed in \cite{Martini 2003} where the corresponding solution of WDVV equations was explicitly stated for the constant multiplicity function.
It was found for any multiplicity function for the non-reduced root systems in \cites{Shen 2018, Shen 2019}.

It follows that a positive half ${\mathcal A} = {\mathcal R}^+$ of a root system $\mathcal R$ also defines a solution of WDVV equations \eqref{WDVV with y}.
We find the corresponding form $G_{\mathcal{R^+}}^{(2)}$  for the root system ${\mathcal R} = BC_{N}$ explicitly in section \ref{section.BCn}.
We also specify corresponding constants $\lambda =\lambda_{(\mathcal{R},c)}$ for (the positive halves of) reduced root systems $\mathcal{R}$ in section \ref{section.root systems solutions revisited}. 
Note that $\lambda$ is invariant under the linear transformations applied to $\mathcal{A}$. In the root system case the scalar $\lambda_{(\mathcal{R},c)}$ may be thought of as a version of the (generalized) Coxeter number for the case of the representation $\Lambda^{2}V$, as the usual (generalized) Coxeter number can also be given as a ratio of two $\mathcal{W}$-invariant forms on $V$ (\cites{Bourbaki, Misha 2012}).    

\section{subsystems of trigonometric $\vee$-systems}\label{section.subsystems of trig.systems}
In this section we consider subsystems of trigonometric $\vee$-systems and show that they are also trigonometric $\vee$-systems. An analogous statement for the rational case was shown in \cite{Misha&Veselov 2008} (see also \cite{Misha 2006}). 
%
 %

A subset $\mathcal{B}\subset\mathcal{A}$ is called a \textit{subsystem} if $\mathcal{B}=\mathcal{A} \cap W$ for some linear subspace $W\subset V^{\ast}.$ 
The subsystem $\mathcal{B}$ is called \textit{reducible} if $\mathcal{B}$ is a disjoint union of two non-empty subsystems, and it is called \textit{irreducible} otherwise. 
Consider the following bilinear form on $V$ associated with a subsystem $\mathcal{B}$:
$$G_{\mathcal{B}}(u,v)\coloneqq \sum_{\beta\in\mathcal{B}}c_{\beta}\beta(u)\beta(v),\quad u,v\in V.$$
 The subsystem $\mathcal{B}$ is called \textit{isotropic} if the restriction $G_{\mathcal{B}}|_{W^{\vee}}$ of the form $G_{\mathcal{B}}$ onto the subspace $W^{\vee}\subset V,$ where $W=\langle	\mathcal{B}\rangle,$
is degenerate and $\mathcal{B}$ is called \textit{non-isotropic} otherwise.  

Let us prove some lemmas which will be useful for the proof of the main theorem of this section.
\begin{lemma}\label{decomposition of W}
Let $\mathcal{A}$ be a trigonometric $\vee$-system. Let $\mathcal{B}=\mathcal{A}\cap W$ be a subsystem of $\mathcal{A}$ for some linear subspace $W \subset V^{\ast}$ such that $W=\langle	\mathcal{B}\rangle$. Consider the linear operator $M\colon V \to W^{\vee}$ given by
\begin{equation}\label{identityoperator}
M=\sum_{\beta\in \mathcal{B}}c_{\beta}\beta\otimes\beta^{\vee},
\end{equation}
that is, $M(v)=\displaystyle\sum_{\beta\in \mathcal{B}}c_{\beta}\beta(v)\beta^{\vee},$ for any $v\in V.$ Then
\begin{enumerate}
\item
\noindent
For any $u, v \in V$ we have 
$G_{\mathcal{A}}(u, M(v))=G_{\mathcal{B}}(u, v).$
\item
For any $\alpha\in \mathcal{B},$ $\alpha^{\vee}$ is an eigenvector for $M.$ 
\item
The space $W^{\vee}$ can be decomposed as a direct sum
\begin{equation}\label{W-decomposition}
W^{\vee}=U_{\lambda_{1}}\oplus U_{\lambda_{2}}\oplus\dots\oplus U_{\lambda_{k}},\quad k\in \mathbb{N},
\end{equation}
where $\lambda_{i}\in \mathbb{C}$ are distinct, and the restriction $ M|_{U_{\lambda_{i}}}=\lambda_{i}I,$ where $I$ is the identity operator. 
\end{enumerate}
\end {lemma}
\begin{proof}
Let $ u,v \in V.$ We have 
$$G_{\mathcal{A}}(u,M(v))=\sum_{\beta \in\mathcal{B}}c_{\beta}\beta (v) G_{\mathcal{A}}(u,\beta^{\vee})=\sum_{\beta\in\mathcal{B}}c_{\beta}\beta (u)\beta(v)=G_{\mathcal{B}}(u,v),$$
which proves the first statement. 

Let us consider a two-dimensional plane $\pi \subset V^{\ast}$ such that $\pi$ contains $\alpha$ and another covector from $\mathcal{B}$ which is not collinear with $\alpha.$ Let us sum up $\vee$-conditions (\ref{Trig.condition}) over
$\alpha$-series which belong to the plane $\pi.$ We get that 
%
$$\sum_{\beta\in \pi\cap\mathcal{A}}c_{\beta}\alpha(\beta^{\vee})\alpha \wedge \beta=\sum_{\beta\in \pi\cap\mathcal{B}}c_{\beta}\beta(\alpha^{\vee})\alpha \wedge \beta=0,$$
hence 
\begin{equation}\label{sum over V-condition 2}
\sum_{\beta\in \pi\cap\mathcal{A}}c_{\beta}\beta(\alpha^{\vee}) \beta^{\vee}=\lambda_{\pi}\alpha^{\vee}
\end{equation}
for some $\lambda_{\pi}\in \mathbb{C}.$
Let us now sum up relation (\ref{sum over V-condition 2}) over all such two-dimensional planes $\pi$ which contain $\alpha$ and another non-collinear covector from $\mathcal{B}.$
It follows that 
$M(\alpha^{\vee})=\lambda \alpha^{\vee},$
for some $\lambda\in\mathbb{C},$ hence property (2) holds.

The set of vectors $\{\alpha^{\vee}\colon\alpha\in \mathcal{B} \}$ spans $W^{\vee}$ since $\mathcal{B}$ spans $W$. As $\alpha^{\vee}$ is an eigenvector for $M|_{W^{\vee}}$ for any $\alpha\in\mathcal{B}$ we get that $M|_{W^{\vee}}$ is diagonalizable, and $W^{\vee}$ has the eigenspace decomposition as stated in (\ref{W-decomposition}).
\end{proof}
\begin{lemma}\label{proportionality}
Let $\mathcal{A}$ and $\mathcal{B}$ be as stated in Lemma \ref{decomposition of W}. Suppose that $\mathcal{B}$ is non-isotropic. Then
\begin{equation}\label{proportionalityof the 2form}
G_{\mathcal{B}}|_{U_{\lambda_{i}}\times V}=\lambda_{i}G_{\mathcal{A}}|_{U_{\lambda_{i}}\times V},
\end{equation}
where $\lambda_{i}\neq 0 $ for all $i=1, \dots, k.$
\end{lemma}
\begin{proof}
Let $u \in V$ and $v \in U_{\lambda_{i}}$ for some $i,$ where $U_{\lambda_{i}}$ is given by (\ref{W-decomposition}). Then by Lemma \ref{decomposition of W} we have
%
$$G_{\mathcal{A}}(u,M(v))=\lambda_{i} G_{\mathcal{A}}(u,v)=G_{\mathcal{B}}(u,v).$$
%
Hence we have the required relation (\ref{proportionalityof the 2form}).
Note that $\lambda_{i}\neq 0$ for all $i$ as otherwise $G_{\mathcal{B}}|_{U_{\lambda_{i}}\times V}=0$ which contradicts the non-isotropicity of $\mathcal{B}.$ 
\end{proof}

Assume that the subsystem $\mathcal{B}=\mathcal{A} \cap W$, $ W=\langle \mathcal{B} \rangle,$ is non-isotropic so that the bilinear form $G_{\mathcal{B}}|_{W^{\vee}}$ is nondegenerate. 
Then it establishes an isomorphism $\phi_{\mathcal{B}}\colon W^{\vee}\to (W^{\vee})^{\ast}.$ 
For any $\beta \in \mathcal{B},$ 
let us denote $\phi_{\mathcal{B}}^{-1}(\beta|_{W^{\vee}})$ by $\beta^{\vee_{\mathcal{B}}}.$
The following lemma relates vectors $\beta^{\vee_{\mathcal{B}}}$ and $\beta^{\vee}.$
\begin{lemma}\label{Proportionality 2}
In the assumptions and notations of Lemmas \ref{decomposition of W} and \ref{proportionality} let $\beta \in \mathcal{B}$.
Let $i\in \mathbb{N}$ be such that $\beta^{\vee}\in U_{\lambda_{i}}.$
Then
$\beta^{\vee_{\mathcal{B}}}=\lambda_{i}^{-1}\beta^{\vee}.$
\end{lemma} 
\begin{proof}
Let $u \in W^{\vee}.$ By Lemma \ref{proportionality} we have
%
$G_{\mathcal{B}}(\beta^{\vee},u)=\lambda_{i}\beta(u)$.
%
By the definition of $\beta^{\vee_{\mathcal{B}}}$ we have
%
$G_{\mathcal{B}}(\beta^{\vee_{\mathcal{B}}},u)=\beta(u)$.
%
%
It follows that $G_{\mathcal{B}}(\lambda_{i}^{-1}\beta^{\vee}-\beta^{\vee_{\mathcal{B}}},u)=0$, which implies
the statement since the form $G_{\mathcal{B}}$ is non-degenerate on $W^{\vee}.$
\end{proof}
\begin{lemma}\label{series in B}
Let $\mathcal{A}$ and $\mathcal{B}$ be as stated in Lemma \ref{decomposition of W}.
Let $\alpha \in \mathcal{B}$ and let $i \in \mathbb{N}$ be such that
$\alpha^{\vee}\in U_{\lambda_{i}}.$ 
Consider an $\alpha$-series $\Gamma_{\alpha}^{\mathcal{B}}$ in $\mathcal{B}$ and let $\beta \in \Gamma_{\alpha}^{\mathcal{B}}.$
Then  $\Gamma_{\alpha}^{\mathcal{B}}\subset U_{\lambda_{i}}$ or 
$\Gamma_{\alpha}^{\mathcal{B}} \subseteq \{\pm \beta\}.$ 
\end{lemma}
\begin{proof}
Suppose firstly that $\beta^{\vee}\in U_{\lambda_{i}}.$ Since any covector $\gamma \in \Gamma_{\alpha}^{\mathcal{B}} $ is a linear combination of $\beta$ and $\alpha$, we get that $\gamma \in  U_{\lambda_{i}}$ as required.

Suppose now that $\beta^{\vee}\notin U_{\lambda_{i}}.$ 
Then $\beta^{\vee}\in U_{\lambda_{j}}$ for some $j\neq i.$
Since we have a direct sum decomposition (\ref{W-decomposition}) it follows that $\Gamma_{\alpha}^{\mathcal{B}} \subseteq \{\pm \beta\}.$ 
\end{proof}
\begin{lemma}\label{seriesrelations}
Let $\mathcal{A}\subset V^{\ast}$ be a finite collection of vectors, and let $\mathcal{B}\subset \mathcal{A}$ be a subsystem.
Let $\alpha,\beta\in \mathcal{B}.$ Let $\Gamma_{\alpha}^{\mathcal{A}},\Gamma_{\alpha}^{\mathcal{B}}$ be the $\alpha$-series in $\mathcal{A}$ and $\mathcal{B}$ respectively containing $\beta.$ Then the set $\Gamma_{\alpha}^{\mathcal{A}}$ coincides with the set $\Gamma_{\alpha}^{\mathcal{B}}.$ 
\end{lemma}
\begin{proof}
Let $\gamma \in \Gamma_{\alpha}^{\mathcal{A}}.$ It follows that $\gamma \in \mathcal{B}.$ 
By maximality of $\Gamma_{\alpha}^{\mathcal{B}},$ it follows that $\gamma\in \Gamma_{\alpha}^{\mathcal{B}}.$ Hence $\Gamma_{\alpha}^{\mathcal{A}}\subset \Gamma_{\alpha}^{\mathcal{B}}.$
The opposite inclusion is obvious.
\end{proof}
\begin{proposition}\label{G_B=0 decompos}
 In the assumptions and notations of Lemma \ref{decomposition of W} we have
%
$G_{\mathcal{B}}(u,v)=0$
%
for any $u\in U_{\lambda_{i}}$ and $v\in U_{\lambda_{j}}$ such that $i \neq j.$
\end{proposition}
\begin{proof}
From Lemma \ref{proportionality} we have
%
$G_{\mathcal{B}}(u,v)=\lambda_{i}G_{\mathcal{A}}(u,v)=\lambda_{j}G_{\mathcal{A}}(u,v)$.
%
Hence $G_{\mathcal{A}}(u,v)=0,$ which implies the statement. 
\end{proof}
Now we present the main theorem of this section.
\begin{theorem}\label{isotropic v-sys. is v-sys.}
Any non-isotropic subsystem of a trigonometric $\vee$-system is also a trigonometric $\vee$-system.
\end{theorem}
\begin{proof}
Let $\mathcal{A}$ be a trigonometric $\vee$-system and let  $\mathcal{B}$ be its non-isotropic subsystem. 
Let $\alpha\in \mathcal{B}.$
Then  $\alpha^{\vee}\in U_{\lambda_{i}}$ in the decomposition (\ref{decomposition of W}) for some $i.$
Consider an $\alpha$-series $\Gamma_{\alpha}^{\mathcal{B}}$ in $\mathcal{B}.$ Let $\beta\in \Gamma_{\alpha}^{\mathcal{B}}.$ Then by Lemma \ref{series in B} we have the following two cases.

\textit{(i)} Suppose $\beta^{\vee}\in U_{\lambda_{i}}.$ Then $\Gamma_{\alpha}^{\mathcal{B}}\subset U_{\lambda_{i}}$ and by Lemmas \ref{proportionality}, \ref{Proportionality 2} we have 
$$G_{\mathcal{B}}(\alpha^{\vee_{\mathcal{B}}},\beta^{\vee_{\mathcal{B}}})=\lambda_{i}^{-2}G_{\mathcal{B}}(\alpha^{\vee},\beta^{\vee})=\lambda_{i}^{-1}G_{\mathcal{A}}(\alpha^{\vee},\beta^{\vee}).$$
Hence we have 
$$\sum_{\beta\in \Gamma_{\alpha}^{\mathcal{B}}}c_{\beta}G_{\mathcal{B}}(\alpha^{\vee_{\mathcal{B}}},\beta^{\vee_{\mathcal{B}}})\alpha\wedge\beta=\lambda_{i}^{-1}\sum_{\beta\in \Gamma_{\alpha}^{\mathcal{B}}} c_{\beta}G_{\mathcal{A}}(\alpha^{\vee},\beta^{\vee})\alpha\wedge\beta=0$$
by Lemma \ref{decomposition of W} and since $\mathcal{A}$ is a trigonometric $\vee$-system. Hence the $\vee$-condition (\ref{Trig.condition}) for $\mathcal{B}$ holds.

\textit{(ii)} Suppose $\beta^{\vee}\in U_{\lambda{j}},$ where $j\neq i.$ Then 
$G_{\mathcal{B}}(\alpha^{\vee_{\mathcal{B}}},\beta^{\vee_{\mathcal{B}}})=\lambda_{i}^{-1}\lambda_{j}^{-1}G_{\mathcal{B}}(\alpha^{\vee},\beta^{\vee})=0,$
%
by Proposition \ref{G_B=0 decompos}, and $\Gamma_{\alpha}^{\mathcal{B}} \subseteq \{ \pm \beta \}$ by Lemma \ref{series in B}. 
Hence the $\vee$-condition (\ref{Trig.condition}) for $\mathcal{B}$ holds.
\end{proof}

\section{restriction of trigonometric solutions of WDVV equations}\label{section.restriction of trig.systems}

In this section we consider the restriction operation for the trigonometric solutions of WDVV equations and show that this gives new solutions of WDVV equations. An analogous statement in the rational case was established in \cite{Misha&Veselov 2007}. 

Let 
\begin{equation}\label{the subsystem B}
    \mathcal{B}=\mathcal{A} \cap W
\end{equation}
be a subsystem of $\mathcal{A}$ for some linear subspace $ W=\langle \mathcal{B} \rangle \subset V^{\ast}.$ 
Define
\begin{equation}\label{W_B the subspace}
W_{\mathcal{B}}=\{x \in V \colon \beta(x)=0 \quad \forall \beta \in \mathcal{B}\}.   
\end{equation}

Let us denote the restriction $\alpha|_{W_{\mathcal{B}}}$ of a covector  $\alpha \in V^{\ast}$ as $\pi_{\mathcal{B}}(\alpha),$ 
 then 
 $$\pi_{\mathcal{B}}(\mathcal{A})=\{\pi_{\mathcal{B}}(\alpha)\colon \pi_{\mathcal{B}}(\alpha)\neq 0, \quad \alpha \in \mathcal{A}\setminus \mathcal{B}\}$$ 
is the set of non-zero restrictions of covectors
$\alpha \in \mathcal{A}$ on $W_{\mathcal{B}}.$ 
Define $M_{\mathcal{B}}=W_{\mathcal{B}} \setminus \bigcup_{\alpha \in \mathcal{A}\setminus \mathcal{B}} \Pi_{\alpha}.$ 

Consider a point $x_{0} \in M_{\mathcal{B}}$ and tangent vectors $u_{0}, v_{0} \in T_{x_{0}}M_{\mathcal{B}}.$ 
We extend vectors $u_{0}$ and $v_{0}$ to two local analytic vector fields $u(x), v(x)$ in the neighbourhood $U$ of $x_{0}$ that are tangent to the subspace $W_{\mathcal{B}} $ at any point $x \in W_{\mathcal{B}} \cap U$ such that $u_{0}=u(x_{0})$ and  $v_{0}=v(x_{0})$. 
Consider the multiplication $\ast$ given by (\ref{a*b in V+U}).
We want to study the limit of $u(x) \ast v(x)$ when $x$ tends to $x_{0}.$
The limit may have singularities at $x \in W_{\mathcal{B}} $ as $\cot \alpha(x)$ with $\alpha \in \mathcal{B} $ is not defined for such $x.$ Also we note that outside $W_{\mathcal{B}}$ we have a well-defined multiplication $u(x) \ast v(x).$

The proof of the next lemma is similar to the proof of  \cite{Misha&Veselov 2007}*{Lemma 1}
in the rational case (see also \cite{MGM 2020}).
\begin{lemma}\label{the limit of *}
The limit of the product $u(x) \ast v(x)$ exists when vector $x$ tends to $x_{0}\in  M_{\mathcal{B}} $ and it satisfies
\begin{equation} \label{limit formula}
u_{0} \ast v_{0} =\sum_{\alpha \in \mathcal{A}\setminus \mathcal{B}} c_{\alpha}\alpha(u_{0})\alpha(v_{0})(\frac{\lambda}{2} \cot \alpha (x_{0})\alpha^{\vee}+E).
\end{equation}
In particular, the product $u_{0} \ast v_{0}$ is determined by vectors $u_{0}$ and $v_{0}$ only. 
\end{lemma}

Now for the subsystem $\mathcal{B}\subset \mathcal{A}$  given by (\ref{the subsystem B}) let 
\begin{equation}\label{Basis of the subsystem B} 
S=\{\alpha_1, \dots, \alpha_k \}\subset \mathcal{B},
\end{equation}
where $k=\dim W,$ be a basis of $W.$ 
The following lemma shows that multiplication (\ref{limit formula}) is closed on the tangent space $T_{\ast}(M_{\mathcal{B}}\oplus U).$
\begin{lemma}\label{closed algebra on V+U}
Let $\mathcal{B}\subset\mathcal{A}$ be a subsystem.
Assume that prepotential (\ref{F with extra variable y}) corresponding to a configuration $(\mathcal{A},c)$ satisfies WDVV equations (\ref{WDVV with y}). 
Suppose that  $C_{\delta_{\alpha}}^{\alpha_{0}}\neq 0$ for any 
${\alpha \in S, \alpha_{0}\in \delta_{\alpha}}.$
If $u,v \in T_{(x,y)}(M_{\mathcal{B}} \oplus U),$ where $x \in W_{\mathcal{B}},y\in U,$ then one has ${u \ast v \in T_{(x,y)}(M_{\mathcal{B}} \oplus U)},$ 
that is 
$$\ast : T_{(x,y)}(M_{\mathcal{B}} \oplus U) \times T_{(x,y)}(M_{\mathcal{B}} \oplus U)\rightarrow T_{(x,y)}(M_{\mathcal{B}} \oplus U),$$
where multiplication $\ast$ is given by (\ref{limit formula}).
\end{lemma}
\begin{proof}
Suppose that the subspace $W_{\mathcal{B}}$ given by (\ref{W_B the subspace}) has codimension $1$ in $V,$ and let $\alpha \in S.$ We have $\mathcal{B}=\delta_{\alpha}.$
Let $x \in M_{\mathcal{B}}\subset \Pi_{\alpha.}$
Let $u,v \in T_{(x,y)}(\Pi_{\alpha} \oplus U)$. Then $u$ and $v$ can be written as 
$u = a_{u}\overline{u}+b_{u}E$, $v  =a_{v}\overline{v}+b_{v}E$, 
where $\overline{u}, \overline{v} \in \Pi_{\alpha},$ and $a_{u},b_{u},a_{v},b_{v} \in \mathbb{C}.$  
By Proposition \ref{identity for A with collinear vectors} we have  
\begin{equation}\label{67}
\sum_{\beta \in \mathcal{A} \setminus \delta_{\alpha}} c_{\beta}G_{\mathcal{A}}(\alpha^{\vee},\beta^{\vee})\cot \beta(x)\alpha(z)\beta(\overline{u})\alpha(w)\beta(\overline{v})=0
\end{equation}
for any $z,w \in V.$ By taking $z,w \notin \Pi_{\alpha}$ we derive from (\ref{67}) that

$$\sum_{\beta \in \mathcal{A}\setminus \mathcal{B}} c_{\beta}\alpha(\beta^{\vee})\beta(\overline{u})\beta(\overline{v})\cot \beta(x)=0,$$ 
which implies the statement by Lemma \ref{the limit of *}.

Let us now consider $W_{\mathcal{B}}$ of codimension $2.$
Let $S=\{\alpha_1 , \alpha_2  \}.$
By the above arguments
$$u\ast v \in T_{(x,y)}(\Pi_{\alpha_{i}} \oplus U)$$
if $x \in \Pi_{\alpha_{i}} $ is generic and $u,v \in T_{(x,y)}(\Pi_{\alpha_{i}} \oplus U),\quad (i=1,2).$
By Lemma \ref{the limit of *}, $u \ast v$ exists for $x \in M_{\mathcal{B}} $ and hence  $u \ast v \in T_{(x,y)}\big((\Pi_{\alpha_{1}}\cap \Pi_{\alpha_{2}}) \oplus U \big).$ 
This proves the statement for the case when $W_{\mathcal{B}}$ has codimension $2.$ General $\mathcal{B}$ is dealt with similarly.
\end{proof}
Let us assume that $G_{\mathcal{A}}|_{W_{\mathcal{B}}}$ is non-degenerate. Then we have the orthogonal decomposition 
%
$$V= W_{ \mathcal{B}} \oplus W_{ \mathcal{B}}^\bot .$$
%
Vector $\alpha^{\vee}\in V$ can be represented as 
\begin{equation} \label{69}
\alpha^{\vee}= \widetilde{\alpha^{\vee}} + w ,
\end{equation}
where $ \widetilde{\alpha^{\vee}} \in  W_{ \mathcal{B}}$ 
and $w \in  W_{ \mathcal{B}}^\bot. $ 
By Lemmas \ref{the limit of *}, \ref{closed algebra on V+U} we have associative product
%
$$u \ast v=\sum_{\alpha \in \mathcal{A}\setminus \mathcal{B}} c_{\alpha}\alpha(u)\alpha(v)(\frac{\lambda}{2} \cot \alpha (x_{0})\widetilde {\alpha^{\vee}}+E),$$
%
where $x_{0}\in M_{\mathcal{B}}, u,v \in W_{\mathcal{B}}.$

For any $\gamma \in W_{\mathcal{B}}^{\ast}$ we define $\gamma^{\vee_{W_\mathcal{B}}}\in W_{\mathcal{B}}$ by
$G_{\mathcal{A}}(\gamma^{\vee_{W_\mathcal{B}}},v)=\gamma(v), \quad \forall v\in W_{\mathcal{B}}.$
\begin{lemma}\label{dual vector under G_B}
Suppose that the restriction $G_{\mathcal{A}}|_{W_{\mathcal{B}}}$ is non-degenerate. Then
 $\widetilde {\alpha^{\vee}}=\pi_{\mathcal{B}}(\alpha)^{\vee_{W_{\mathcal{B}}}}$ for any $\alpha \in V^{\ast}.$
\end{lemma}
\begin{proof}
From decomposition (\ref{69}) we have
$$\alpha (v) =G_{\mathcal{A}} (\alpha^{\vee},v)= G_{\mathcal{A}}(\widetilde{\alpha^{\vee}}+w,v)=G_{\mathcal{A}} (\widetilde{\alpha^{\vee}},v)$$
for any $\gamma \in W_{\mathcal{B}}.$
It is follows that 
$G_{\mathcal{A}}(\pi_{\mathcal{B}}(\alpha)^{\vee_{W_{\mathcal{B}}}}-\widetilde{\alpha^{\vee}},v)=0,$ 
which implies the statement as $G_{\mathcal{A}}|_{W_{\mathcal{B}}}$ is non-degenerate.
\end{proof}

Let us choose a basis in the space $W_{\mathcal{B}} \oplus U$ such that $f_1,\dots,f_n$ is a basis in $W_{\mathcal{B}}, n=\dim W_{\mathcal{B}},$ and $f_{n+1}$ is the basis vector in $U,$
and let $\xi_1, \dots,  \xi_{n+1}$ be the corresponding coordinates.
We represent vectors $\xi\in W_{\mathcal{B}}, y\in U$ as $\xi=(\xi_{1},...,\xi_{n})$ and $y=\xi_{n+1}.$ 
The WDVV equations for a function $F \colon W_{\mathcal{B}} \oplus U \to \mathbb{C}$ is the following system of partial differential equations:
\begin{equation}\label{WDVV with y in W_B}
\ F_{i}F_{n+1}^{-1}F_{j}=F_{j}F_{n+1}^{-1}F_{i},\quad i,j=1,...,n,
\end{equation}
where $F_{i}$ is $(n+1)\times(n+1)$ matrix with entries
$(F_i)_{p q}=\frac{\partial^{3}F}{\partial \xi_{i}\partial
\xi_{p}\partial \xi_{q}}$ ($p,q =1,\dots, n+1$). 
The previous considerations lead to the following theorem.
\begin{theorem} \label{restricted system and WDVV}
Let $\mathcal{B}\subset\mathcal{A}$ be a subsystem, and let $S$ be as defined in (\ref{Basis of the subsystem B}). 
Assume that prepotential (\ref{F with extra variable y}) satisfies WDVV equations (\ref{WDVV with y}). 
Suppose that  $C_{\delta_{\alpha}}^{\alpha_{0}}\neq 0$ for any $\alpha \in S, \alpha_{0}\in \delta_{\alpha}.$
Then the prepotential 
\begin{equation} \label{F_B}
F_{\mathcal{B}}=F_{\mathcal{B}}(\xi,y)=\frac{1}{3}y^3+\sum_{\alpha \in \mathcal{A}\setminus \mathcal{B}}c_{\alpha} \overline{\alpha}(\xi)^{2}y+\lambda\sum_{\alpha \in \mathcal{A}\setminus \mathcal{B}}c_{\alpha}f(\overline{\alpha}(\xi)),\quad \xi \in W_{\mathcal{B}},y \in  U\cong \mathbb{C},
\end{equation}
where $\overline{\alpha}=\pi_{ \mathcal{B}}(\alpha),$ 
satisfies the WDVV equations (\ref{WDVV with y in W_B}). The corresponding associative multiplication has the form 
\begin{equation} \label{algebra on restrected system}
u \ast v =\sum_{\alpha \in \mathcal{A}\setminus \mathcal{B}} c_{\alpha}\overline{\alpha}(u)\overline{\alpha}(v)(\frac{\lambda}{2} \cot \overline{\alpha} (\xi)\overline{\alpha}^{\vee_{W_{ \mathcal{B}}}}+E),
\end{equation}
where $ \xi \in M_{\mathcal{B}},u,v\in T_{(\xi,y)}M_{\mathcal{B}}.$
\end{theorem}
\begin{proof}
It follows by Lemmas \ref{associyivity and WDVV}, \ref{the limit of *}--\ref{dual vector under G_B}, that multiplication (\ref{algebra on restrected system}) is associative. The corresponding prepotential has the form (\ref{F_B}) and it satisfies WDVV equations (\ref{WDVV with y in W_B}) by Lemma \ref{associyivity and WDVV}.
\end{proof}
In general a restriction of a root system is not a root system, so we get new solutions of WDVV equations by applying Theorem \ref{restricted system and WDVV} in this case. In sections \ref{section.BCn}, \ref{section.An} and \ref{section.root systems solutions revisited} we consider such solutions in more details.
\section{$BC_N$ type configurations}\label{section.BCn}

In this section we discuss a family of configurations of $BC_{N}$ type and show that it gives trigonometric solutions of the WDVV equations.
Let the set $\mathcal{A}=BC_N^{+}$ consist of the following covectors:
$$e^i,  2e^i,\quad (1 \leq i \leq N), \quad e^i \pm e^j, (1\leq i < j \leq N).$$
Let us define the multiplicity function $c\colon BC_N^{+} \to \mathbb{C}$ by $c(e^i)=r$, $ c(2 e^i)=s$, $c(e^i \pm e^j)=q,$ where $r, s, q \in \mathbb{C}.$
We will denote the configuration $(BC_N^{+},c)$ as $BC_N^{+}(r,s,q).$
It is easy to check that 
\begin{equation}\label{bilinear form of BCn}
    G_\mathcal{A}(u,v)= {h} \langle u,v \rangle,  \quad u,v \in V,
\end{equation}
where
\begin{equation}\label{factor of bilinear BCn}
h =r+4s+2q(N-1) 
\end{equation}
is assumed to be non-zero,
and $ \langle u,v \rangle=\sum_{i=1}^{N}u_{i}v_{i}$ is the standard inner product for
${u=(u_{1},\dots,u_{N})}$, ${v=(v_{1},\dots,v_{N})}$.
For any $\alpha, \beta \in V^{\ast},$ 
$(\alpha \wedge \beta)^{2} \colon V \otimes V \to \mathbb{C}$ denotes the square of the covector $\alpha \wedge \beta \in (V \otimes V)^{\ast}.$
\begin{lemma} \label{IDENTITIES.BCn}
The following two identities hold:

\begin{equation}\label{summation identity 1}
\displaystyle\sum_{1\leq i < j<k \leq N}\bigg((e^{i}\wedge e^{j})^{2}+(e^{i}\wedge e^{k})^{2}+(e^{j}\wedge e^{k})^{2}\bigg)=(N-2)\displaystyle \sum_{1\leq i < j \leq N}(e^{i}\wedge e^{j})^{2}
\end{equation}
and
\begin{align}\label{summation identity 2}
&\displaystyle \sum_{1\leq i < j<k<l \leq N}\bigg((e^{i}\wedge e^{j})^{2}+(e^{i}\wedge e^{k})^{2}+(e^{i}\wedge e^{l})^{2}+(e^{j}\wedge e^{k})^{2}+(e^{j}\wedge e^{l})^{2}+(e^{k}\wedge e^{l})^{2}\bigg)\nonumber\\
&=\frac{1}{2}(N-2)(N-3)\displaystyle \sum_{1\leq i < j \leq N}(e^{i} \wedge e^{j})^{2}.
\end{align}
\end{lemma}  
\begin{proof}
Note that
\begin{equation}\label{sum1}
\sum_{1\leq i < j<k \leq N}(e^{i}\wedge e^{j})^{2}=\sum_{1\leq i < j \leq N}(N-j)(e^{i}\wedge e^{j})^{2},
\end{equation}
\begin{equation}\label{sum2}
\sum_{1\leq i < j<k \leq N}(e^{i}\wedge e^{k})^{2}=\sum_{1\leq i < k \leq N}(k-i-1)(e^{i}\wedge e^{k})^{2},
\end{equation}
and
\begin{equation}\label{sum3}
\sum_{1\leq i < j<k \leq N}(e^{j}\wedge e^{k})^{2}=\sum_{1\leq j < k\leq N}(j-1)(e^{j}\wedge e^{k})^{2}.
\end{equation}
By adding together relations (\ref{sum1})--(\ref{sum3}) we get identity (\ref{summation identity 1}).

We also have
\begin{equation}\label{sum1.1}
\displaystyle \sum_{1\leq i < j<k<l \leq N}(e^{i}\wedge e^{j})^{2}=\sum_{1\leq i < j \leq N}\frac{1}{2}(N-j-1)(N-j)(e^{i}\wedge e^{j})^{2},
\end{equation}
\begin{equation}\label{sum1.2}
\displaystyle \sum_{1\leq i < j<k<l \leq N}(e^{i}\wedge e^{k})^{2}=\sum_{1\leq i < k \leq N}(N-k)(k-i-1)(e^{i}\wedge e^{k})^{2},
\end{equation}
\begin{equation}\label{sum1.3}
\displaystyle \sum_{1\leq i < j<k<l \leq N}(e^{i}\wedge e^{l})^{2}=\sum_{1\leq i < l \leq N}\frac{1}{2}(l-i-2)(l-i-1)(e^{i}\wedge e^{l})^{2},
\end{equation}
\begin{equation}\label{sum1.4}
\displaystyle \sum_{1\leq i < j<k<l \leq N}(e^{j}\wedge e^{k})^{2}=\sum_{1\leq i < j \leq N}(N-j)(i-1)(e^{i}\wedge e^{j})^{2},
\end{equation}
\begin{equation}\label{sum1.5}
\displaystyle \sum_{1\leq i < j<k<l \leq N}(e^{j}\wedge e^{l})^{2}=\sum_{1\leq j < l \leq N}(l-j-1)(j-1)(e^{j}\wedge e^{l})^{2},
\end{equation}
and
\begin{equation}\label{sum1.6}
\displaystyle \sum_{1\leq i < j<k<l \leq N}(e^{k}\wedge e^{l})^{2}=\sum_{1\leq k < l \leq N}\frac{1}{2}(k-2)(k-1)(e^{k}\wedge e^{l})^{2}.
\end{equation}
Then by adding together identities (\ref{sum1.1})--(\ref{sum1.6}) we obtain identity (\ref{summation identity 2}).

\end{proof}
\begin{proposition}\label{identitiesforlambdaBCn}
The quadratic forms $G_{\mathcal{A}}^{(1)}, G_{\mathcal{A}}^{(2)}$ corresponding to the bilinear forms $G_{\mathcal{A}}^{(1)}(\cdot, \cdot)$, $G_{\mathcal{A}}^{(2)}(\cdot, \cdot)$ respectively have the following forms:
\begin{equation}\label{summation identity 4}
G_{\mathcal{A}}^{(1)}=
2 {h}^{2}\displaystyle \sum_{1\leq i < j \leq N}(e^{i}\wedge e^{j})^{2},
\end{equation}
and
\begin{equation}\label{summation identity 3}
G_{\mathcal{A}}^{(2)}=
4  q \big(r+8s+2(N-2)q\big) {h}^{-1} \displaystyle \sum_{1\leq i < j \leq N}(e^{i}\wedge e^{j})^{2},
\end{equation}
where $h$ is given by (\ref{factor of bilinear BCn}).
\end{proposition}
\begin{proof}
Let us first prove identity (\ref{summation identity 4}).
Note that  $G_{\mathcal{A}}^{(1)}$ is a quadratic polynomial in $r,s$ and $q.$ The terms containing $r^{2}$ add up to 
\begin{equation}\label{r^2-terms}
2 r^{2}\displaystyle \sum_{1\leq i < j \leq N}(e^{i}\wedge e^{j})^{2}.
\end{equation}
Similarly, the terms containing $s^{2}$ add up to 
\begin{equation}\label{s^2-terms}
32 s^{2}\displaystyle \sum_{1\leq i < j \leq N} (e^{i}\wedge e^{j})^{2}.
\end{equation}
%
The terms containing $rs$ add up to 
\begin{equation}\label{rs-terms}
2 rs\displaystyle \sum_{1\leq i < j \leq N}\bigg((e^{i}\wedge 2e^{j})^{2}+(2e^{i}\wedge e^{j})^{2}\bigg)=16 rs\displaystyle \sum_{1\leq i < j \leq N}(e^{i}\wedge e^{j})^{2}.
\end{equation}
Now the terms containing $rq$ have the form
\begin{align}\label{rq-terms2}
& 2 rq\sum_{1\leq i < j \leq N}\bigg(\big(e^{i}\wedge(e^{i}+e^{j}\big)\big)^{2}
+\big(e^{i}\wedge(e^{i}-e^{j}\big)\big)^{2}
+\big( e^{j}\wedge(e^{i}+e^{j})\big)^{2}
+\big(e^{j}\wedge(e^{i}-e^{j})\big)^{2}\bigg)\nonumber \\
&+2 rq\sum_{1\leq i < j<k \leq N}\bigg(\big(e^{i}\wedge(e^{j}+e^{k}\big)\big)^{2}
+\big(e^{i}\wedge(e^{j}-e^{k}\big)\big)^{2}
+\big( e^{k}\wedge(e^{i}+e^{j})\big)^{2}
+\big(e^{k}\wedge(e^{i}-e^{j})\big)^{2}\nonumber \\
&+\big( e^{j}\wedge(e^{i}+e^{k})\big)^{2}
+\big( e^{j}\wedge(e^{i}-e^{k})\big)^{2}\bigg)\nonumber \\
&= 8 rq\sum_{1\leq i < j \leq N}(e^{i}\wedge e^{j})^{2} 
+2 rq\sum_{1\leq i < j<k \leq N}\Bigg(\bigg(e^{i}\wedge e^{j}+e^{i}\wedge e^{k}\bigg)^{2}
+\bigg(e^{i}\wedge e^{j}-e^{i}\wedge e^{k}\bigg)^{2}
\nonumber \\
&+\bigg(e^{i}\wedge e^{k}+e^{j}\wedge e^{k}\bigg)^{2}
+\bigg(e^{i}\wedge e^{k}-e^{j}\wedge e^{k}\bigg)^{2}
+\bigg(e^{i}\wedge e^{j}-e^{j}\wedge e^{k}\bigg)^{2}
+\bigg(e^{i}\wedge e^{j}+e^{j}\wedge e^{k}\bigg)^{2}\Bigg)\nonumber\\
&= 8 rq\sum_{1\leq i < j \leq N}(e^{i}\wedge e^{j})^{2}
+8 rq\sum_{1\leq i < j<k \leq N}\bigg((e^{i}\wedge e^{j})^{2}+(e^{i}\wedge e^{k})^{2}+(e^{j}\wedge e^{k})^{2}\bigg)\nonumber\\
&=8 rq(N-1)\sum_{1\leq i < j \leq N}(e^{i}\wedge e^{j})^{2}
\end{align}
by Lemma \ref{IDENTITIES.BCn}.
Similarly the terms containing $sq$ add up to
\begin{equation}\label{sq-terms 2}
32 sq(N-1)\sum_{1\leq i < j \leq N}(e^{i}\wedge e^{j})^{2}.
\end{equation}
The terms containing $q^{2}$ have the form
\begin{align}\label{q^2-terms 2b}
& 2 q^{2}\sum_{1\leq i < j \leq N}\big((e^{i}+e^{j})\wedge (e^{i}-e^{j})^{2}\big)\nonumber\\
&+2 q^{2}\sum_{1\leq i < j<k \leq N}\bigg(\big((e^{i}+e^{j})\wedge(e^{i}+e^{k})\big)^{2}
+\big((e^{i}+e^{j})\wedge(e^{i}-e^{k})\big)^{2}+\big((e^{i}-e^{j})\wedge(e^{i}+e^{k})\big)^{2} \nonumber\\
&+\big((e^{i}-e^{j})\wedge(e^{i}-e^{k})\big)^{2}
+\big((e^{i}+e^{j})\wedge(e^{j}+e^{k})\big)^{2}
+\big((e^{i}+e^{j})\wedge(e^{j}-e^{k})\big)^{2} \nonumber\\
&+\big((e^{i}-e^{j})\wedge(e^{j}+e^{k})\big)^{2}
+\big((e^{i}-e^{j})\wedge(e^{j}-e^{k})\big)^{2}
+\big((e^{i}+e^{k})\wedge(e^{j}+e^{k})\big)^{2}
\nonumber\\
&+\big((e^{i}+e^{k})\wedge(e^{j}-e^{k})\big)^{2}+\big((e^{i}-e^{k})\wedge(e^{j}+e^{k}\big)^{2}
+\big((e^{i}-e^{k})\wedge(e^{j}-e^{k}\big)^{2}\bigg)\nonumber\\
&+2 q^{2}\sum_{1\leq i < j<k<l \leq N}\bigg(\big((e^{i}+e^{j})\wedge(e^{k}+e^{l})\big)^{2}
+\big((e^{i}+e^{j})\wedge(e^{k}-e^{l})\big)^{2} \nonumber\\
&+\big((e^{i}-e^{j})\wedge(e^{k}+e^{l})\big)^{2}
+\big((e^{i}-e^{j})\wedge(e^{k}-e^{l})\big)^{2}
+\big((e^{i}+e^{k})\wedge(e^{j}+e^{l})\big)^{2}\nonumber\\
&+\big((e^{i}+e^{k})\wedge(e^{j}-e^{l})\big)^{2}
+\big((e^{i}-e^{k})\wedge(e^{j}+e^{l})\big)^{2}
+\big((e^{i}-e^{k})\wedge(e^{j}-e^{l})\big)^{2}\nonumber\\
&+\big((e^{i}+e^{l})\wedge(e^{j}+e^{k})\big)^{2}
+\big((e^{i}+e^{l})\wedge(e^{j}-e^{k})\big)^{2}+\big((e^{i}-e^{l})\wedge(e^{j}+e^{k})\big)^{2}\nonumber\\
&+\big((e^{i}-e^{l})\wedge(e^{j}-e^{k})\big)^{2}\bigg).
\end{align}
Expression (\ref{q^2-terms 2b}) is equal to 
\begin{align}\label{q^2-terms 4b}
& 8 q^{2}\sum_{1\leq i < j \leq N}(e^{i} \wedge e^{j})^{2}
+24 q^{2}\sum_{1\leq i < j<k \leq N}\bigg((e^{i}\wedge e^{j})^{2}+(e^{i}\wedge e^{k})^{2}+(e^{j}\wedge e^{k})^{2}\bigg)\nonumber \\
&+16 q^{2}\sum_{1\leq i < j<k<l \leq N}\bigg((e^{i}\wedge e^{j})^{2}+(e^{i}\wedge e^{k})^{2}+(e^{i}\wedge e^{l})^{2}+(e^{j}\wedge e^{k})^{2}+(e^{j}\wedge e^{l})^{2}+(e^{k}\wedge e^{l})^{2}\bigg) \nonumber \\
&=8 q^{2}(N-1)^{2}\sum_{1\leq i < j\leq N}(e^{i}\wedge e^{j})^{2}
\end{align}
by Lemma \ref{IDENTITIES.BCn}.
By adding together expressions (\ref{r^2-terms})--(\ref{sq-terms 2}) and (\ref{q^2-terms 4b}) we get identity (\ref{summation identity 4}).

Let us now prove identity (\ref{summation identity 3}).
Note that ${h} G_{\mathcal{A}}^{(2)}$ is a quadratic polynomial in $r,s$ and $q$ and that terms containing $r^{2},rs$ and $s^{2}$ all vanish.
Terms containing $rq$ in ${h} G_{\mathcal{A}}^{(2)}$ are given by
\begin{align}\label{rq-terms}
& 2rq  \sum_{1\leq i < j \leq N}\bigg( e^{i}(e^{i}+e^{j})^{\vee}\Big(e^{i}\wedge(e^{i}+e^{j})\Big)^{2}
+e^{i}(e^{i}-e^{j})^{\vee}\Big(e^{i}\wedge(e^{i}-e^{j})\Big)^{2} \nonumber \\
&+ e^{j}(e^{i}+e^{j})^{\vee} \Big(e^{j} \wedge(e^{i}+e^{j})\Big)^{2}
+e^{j}(e^{i}-e^{j})^{\vee} \Big(e^{j} \wedge(e^{i}-e^{j})\Big)^{2}\bigg) =4rq \sum_{1\leq i < j \leq N}(e^{i}\wedge e^{j})^{2}.
\end{align}
Similarly, the terms containing $sq$ in $ h G_{\mathcal{A}}^{(2)}$ add up to 
\begin{align}\label{sq-terms}
32sq\sum_{1\leq i < j \leq N}(e^{i}\wedge e^{j})^{2}.
\end{align}
Finally, the terms containing $q^{2}$ in $ h G_{\mathcal{A}}^{(2)}$ are given by
\begin{align}\label{q^2-terms}
& 2q^{2} \sum_{1\leq i < j<k \leq N}\bigg((e^{i}+e^{j})\big((e^{i}+e^{k})^{\vee}\big)\big((e^{i}+e^{j})\wedge(e^{i}+e^{k}\big)\big)^{2} \nonumber\\
&+(e^{i}+e^{j})\big((e^{i}-e^{k})^{\vee}\big)\big((e^{i}+e^{j})\wedge(e^{i}-e^{k}\big)\big)^{2} + (e^{i}-e^{j})\big( (e^{i}+e^{k})^{\vee}\big)\big((e^{i}-e^{j})\wedge(e^{i}+e^{k}\big)\big)^{2} \nonumber\\
&+(e^{i}-e^{j})\big((e^{i}-e^{k})^{\vee}\big)\big((e^{i}-e^{j})\wedge(e^{i}-e^{k}\big)\big)^{2} + (e^{i}+e^{j})\big((e^{j}+e^{k})^{\vee}\big)\big((e^{i}+e^{j})\wedge(e^{j}+e^{k}\big)\big)^{2} \nonumber\\
&+(e^{i}+e^{j})\big((e^{j}-e^{k})^{\vee}\big)\big((e^{i}+e^{j})\wedge(e^{j}-e^{k}\big)\big)^{2}  + (e^{i}-e^{j})\big((e^{j}+e^{k})^{\vee}\big)\big((e^{i}-e^{j})\wedge(e^{j}+e^{k}\big)\big)^{2} \nonumber\\
&+(e^{i}-e^{j})\big((e^{j}-e^{k})^{\vee}\big)\big((e^{i}-e^{j})\wedge(e^{j}-e^{k}\big)\big)^{2}  + (e^{i}+e^{k})\big((e^{j}+e^{k})^{\vee}\big)\big((e^{i}+e^{k})\wedge(e^{j}+e^{k}\big)\big)^{2} \nonumber\\
&+(e^{i}+e^{k})\big((e^{j}-e^{k})^{\vee}\big)\big((e^{i}+e^{k})\wedge(e^{j}-e^{k}\big)\big)^{2} + (e^{i}-e^{k})\big((e^{j}+e^{k})^{\vee}\big)\big((e^{i}-e^{k})\wedge(e^{j}+e^{k}\big)\big)^{2} \nonumber\\
&+(e^{i}-e^{k})\big((e^{j}-e^{k})^{\vee}\big)\big((e^{i}-e^{k})\wedge(e^{j}-e^{k}\big)\big)^{2}\bigg).
\end{align}
Expression (\ref{q^2-terms}) is equal to
\begin{align}\label{q^2-terms 3}
&2q^{2}\sum_{1\leq i < j<k \leq N}\bigg(\big(e^{i}\wedge e^{k}-e^{i}\wedge e^{j}+ e^{j}\wedge e^{k} \big)^{2}
+\big(e^{i}\wedge e^{k}+e^{i}\wedge e^{j}- e^{j}\wedge e^{k} \big)^{2}\nonumber\\
&+\big(e^{i}\wedge e^{k}- e^{i}\wedge e^{j}- e^{j}\wedge e^{k} \big)^{2}
+\big(e^{i}\wedge e^{j}+e^{i}\wedge e^{k}+ e^{j}\wedge e^{k} \big)^{2}\bigg)\nonumber\\
&=8q^{2}\sum_{1\leq i < j<k \leq N}\bigg((e^{i}\wedge e^{j})^{2}
+(e^{i}\wedge e^{k})^{2}+(e^{j}\wedge e^{k})^{2}\bigg)=8q^{2}(N-2)\sum_{1\leq i < j\leq N}(e^{i}\wedge e^{j})^{2}
\end{align}
by Lemma \ref{IDENTITIES.BCn}. 
By adding together expressions (\ref{rq-terms}), (\ref{sq-terms}) and (\ref{q^2-terms 3}) we get 
identity (\ref{summation identity 3}).
\end{proof}

The previous proposition allows us to prove the following theorem. 
\begin{theorem} \label{lambda for general BC_N}
Prepotential (\ref{F with extra variable y}) for the configuration $(\mathcal{A},c)=BC_{N}^{+}(r,s,q)$ satisfies WDVV equations (\ref{WDVV with y}) with 
\begin{equation}\label{Lambda formula for general BC_N}
\lambda=\Big(\frac{2 {h}^{3}}{q \big(r+8s+2(N-2)q\big)}\Big)^{1/2},
\end{equation}
where $h$ is given by (\ref{factor of bilinear BCn}),
provided that $q(r+8s+2(N-2)q)\neq 0.$ 
\end{theorem}
\begin{proof}
Firstly, $BC_{N}^{+}(r,s,q)$ is a trigonometric $\vee$-system by Proposition \ref{root system is trig. system}. 
Secondly, by Proposition~ \ref{identitiesforlambdaBCn} we have that 
$G_{\mathcal{A}}^{(1)}- \frac{\lambda^2}{4} G_{\mathcal{A}}^{(2)}=0$
if $\lambda$ is given by (\ref{Lambda formula for general BC_N}). The statement follows by Theorem~ \ref{correction of Misha's Theorem}.
\end{proof}  
Theorem \ref{lambda for general BC_N} gives a generalization of the results in \cite{Martini 2003 (1)}, \cite{Martini 2003}, \cite{Bryan 2008} and \cite{Shen 2019}, where, in particular, solutions of the WDVV equations for the root systems $D_N, B_N$ and $C_N$ were obtained.
Following \cite{Martini 2003 (1)}, \cite{Martini 2003} consider the function $\widetilde{F}$ of $N+1$ variables $(x_1,\dots,x_N,y)$ of the form
\begin{equation}\label{F (Mrtini)}
\widetilde{F}(x,y)=\frac{\gamma}{6}y^3+\frac{\gamma}{2}y \langle x, x \rangle +\sum_{\alpha\in\mathcal{R}^{+}}c_{\alpha}\widetilde{f}(\alpha(x)),
\end{equation}
where $\mathcal{R}^{+}$ is a positive half of the root system $\mathcal{R}$, multiplicities $c_{\alpha}$ are invariant under the Weyl group, $\gamma \in \mathbb{C}$ and function $\widetilde{f}$ given by
\begin{equation}\label{f for martini}
    \widetilde{f}(z)=\frac{1}{6} z^3-\frac{1}{4}Li_3(e^{-2z})
\end{equation}
satisfies  $\widetilde{f}'''(z)=\coth z.$ 
Note that $\widetilde{f}(z)=-f(-iz).$

Let us explain that our solution (\ref{F with extra variable y}) for the configuration $BC_{N}^{+}(r,s,q)$ leads to a solution of the form (\ref{F (Mrtini)}). 
\begin{proposition} \label{solution for BC_n of Martini theorem}
Function $\widetilde{F}$ given by (\ref{F (Mrtini)}) with $\mathcal{R}^{+}=BC_{N}^{+}$ satisfies WDVV equations (\ref{WDVV with y}) if
\begin{equation}\label{gamma for BCn}
    \gamma^{2}=-2q(r+8s+2(N-2)q).
\end{equation}
\end{proposition}
\begin{proof}
By formula (\ref{bilinear form of BCn}) solution $F$ given by (\ref{F with extra variable y}) for $\mathcal{A}=BC_N^{+}$ has the form
\begin{equation}\label{F_1 for BC_n}
F(\widetilde{x},\widetilde{y})=\frac{1}{3}\widetilde{y}^3+{h} \widetilde{y} \sum_{i=1}^{N}\widetilde{x}_{i}^{2}+\lambda\sum_{\alpha\in BC_N^{+}}c_{\alpha}f(\alpha(\widetilde{x})),
\end{equation}
where we redenoted variables $(x,y)$ by $(\widetilde{x},\widetilde{y}).$
By changing variables $\widetilde{x}=-ix$, $ \widetilde{y}=\frac{\gamma \lambda }{2 {h}}y$ and dividing by $-\lambda$ solution (\ref{F_1 for BC_n}) takes the form (\ref{F (Mrtini)}) provided that $\gamma^{2}\lambda^{2}=-4{h}^{3}$ which is satisfied for $\gamma$ given by (\ref{gamma for BCn}). 
\end{proof}

Let $n\in\mathbb{N}$ and let $\underline{m}=(m_1, \dots, m_n)\in \mathbb{N}^{n}$ be such that
\begin{equation}\label{sum of partitions}
    \sum_{i=1}^n m_i = N. 
\end{equation}
Let us consider the subsystem $\mathcal{B}\subset \mathcal{A}=BC_{N}^{+}$ given by
$$\mathcal{B}=\{e^{\sum_{j=1}^{i-1}m_{j}+k}-e^{\sum_{j=1}^{i-1}m_{j}+l},\quad 1\leq k<l\leq m_{i},\quad i=1,\dots,n\}.$$
Let us also consider the corresponding subspace $W_{\mathcal{B}}=\{x\in V\colon \beta(x)=0, \forall \beta\in \mathcal{B} \}.$
It can be given explicitly by the equations
\begin{align*}
\begin{dcases}
x_{1}=\dots =x_{m_{1}}=\xi_{1},\\
x_{m_{1}+1}=\dots =x_{m_{1}+m_{2}}=\xi_{2},\\
\vdots\\
x_{\sum_{i=1}^{n-1}m_i +1}=\dots =x_{N}=\xi_{n},
\end{dcases}
\end{align*}
where $\xi_{1}, \dots, \xi_{n}$ are coordinates on $W_{\mathcal{B}}.$   
Let us now restrict the configuration $BC_{N}^{+}(r,s,q)$ to the subspace $W_{\mathcal{B}}.$ 
That is we consider non-zero restricted covectors $\overline{\alpha}=\pi_{\mathcal{B}}(\alpha), \alpha \in BC_{N}^{+}$ with multiplicities $c_{\alpha},$ and we add up multiplicities if the same covector on $W_{\mathcal{B}}$ is obtained a few times.
Let us denote the resulting configuration as 
$BC_n(q, r, s ; \underline{m}).$ It is easy to see that it consists of covectors
\begin{align*}
\noindent f^i, \quad  \text{with multiplicity} \quad r m_i, \quad 1 \leq i \leq n,\\
\noindent 2f^i, \quad \text{with multiplicity} \quad s m_i + \frac{1}{2}qm_i(m_i -1),\quad 1 \leq i \leq n,\\
\noindent \quad f^i \pm f^j, \quad \text{with multiplicity} \quad qm_i m_j, \quad 1\leq i < j \leq n,
\end{align*}
where $f^1, \dots , f^n$ is the basis in $W_{\mathcal{B}}^{\ast}$ corresponding to coordinates $\xi_1 , \dots, \xi_n .$

As a corollary of Theorem \ref{restricted system and WDVV} and Theorem \ref{lambda for general BC_N} we get the following result on $(n+3)$-parametric family of solutions of WDVV equations, which can be specialized to $(n+1)$-parametric family of solutions from \cite{Pavlov 2006}.  
\begin{theorem}
Let $\xi=(\xi_1 , \dots, \xi_n)\in W_{\mathcal{B}},y \in U \cong  \mathbb{C}.$ Assume that parameters $r,q,s$ and $\underline{m}$ satisfy the relation
$  r+4s +2q(m_{i}-1) \neq 0$ for any $1\leq i \leq n.$
Then function 
\begin{align}\label{restricted function for BCn}
 F_{\widetilde{\mathcal{A}}}(\xi,y)&=\frac{1}{3}y^3+ \big(r+4s+2q(N-1) \big)y \sum_{i=1}^{n} m_i \xi_{i}^{2} + \lambda r \sum_{i=1}^{n} m_i f(\xi_i) \nonumber \\
 &+\lambda \sum_{i=1}^{n}\big(sm_i+\frac{1}{2}qm_i(m_i -1) \big)f(2\xi_{i})+\lambda q\sum_{i<j}^{n} m_{i}m_{j}f(\xi_i \pm \xi_j), 
\end{align}
where $N$ is given by (\ref{sum of partitions}),
satisfies the WDVV equations (\ref{WDVV with y in W_B}) if $\lambda=\Big(\frac{2 {h}^{3}}{q\big(r+8s+2(N-2)q \big)}\Big)^{1/2}$, where $h$ is given by \eqref{factor of bilinear BCn}, and 
$\big(r+8s+2(N-2)q \big)q\neq 0.$
\end{theorem}
\begin{proof}
We only have to check that cubic terms in (\ref{restricted function for BCn}) have the required form.
For any $\xi \in W_{\mathcal{B}}$ we have
\begin{align}\label{quadratic terms for BC_N}
    \sum_{\alpha \in BC_{N}^{+}}c_{\alpha} \overline{\alpha}(\xi)^{2}&= r\sum_{i=1}^{n}m_i \xi_{i}^{2}+4\sum_{i=1}^{n} \big(sm_i+\frac{1}{2}qm_i(m_i -1)\big)\xi_{i}^{2}\nonumber \\
    &+2q\sum_{1 \leq i<j\leq n}m_i m_j (\xi_{i}^{2}+\xi_{j}^{2}).
\end{align}
Note that
$$ \sum_{1 \leq i<j\leq n}m_i m_j (\xi_{i}^{2}+\xi_{j}^{2})
=\frac{1}{2} \sum_{i,j=1}^{n}m_i m_j (\xi_{i}^{2}+\xi_{j}^{2})-\sum_{i=1}^{n}m_{i}^{2} \xi_{i}^{2}
= \sum_{i=1}^{n}(N-m_i)m_i \xi_{i}^{2}$$
by formula (\ref{sum of partitions}). Hence (\ref{quadratic terms for BC_N}) becomes
\begin{align*}
   \sum_{\alpha \in BC_{N}^{+}}c_{\alpha} \overline{\alpha}(\xi)^{2}&=\big(r+4s+2q(N-1) \big)  \sum_{i=1}^{n}m_{i} \xi_{i}^{2}
\end{align*}
as required.
\end{proof}

\section{$A_N$ type configurations}\label{section.An}
In this section we discuss a family of configurations of type $A_{N}$ and show that it gives trigonometric solutions of the WDVV equations. 

Let $V\subset \mathbb{C}^{N+1}$ be the hyperplane $V=\{(x_1, \dots, x_{n+1})\colon \sum_{i=1}^{N+1}x_i=0   \}.$
Let $\mathcal{A}=A_{N}^{+}$ be the positive half of the root system $A_N$ given by
$$\mathcal{A}=\{e^i-e^j, \quad 1\leq i<j \leq N+1 \}.$$
Let $t=c(e^i-e^j)\in \mathbb C$ be the constant multiplicity.
The following lemma gives the relation between covectors in $\mathcal{A}$ and their dual vectors in $V$. 
\begin{lemma}\label{vectors and covectors in A_N}
We have 
$$(e^{i}-e^{j})^{\vee}=\frac{1}{t(N+1)}(e_{i}-e_{j}), \quad 1\leq i,j\leq N+1.$$
\end{lemma}
\begin{proof}
 Let $x=(x_1,\dots,x_{N+1}),y=(y_1,\dots,y_{N+1}) \in V.$ 
%
Then the bilinear form $G_{\mathcal{A}}$ takes the form
%
 $$   G_{\mathcal{A}}(x,y)=t \sum_{i<j}^{N+1}(x_{i}-x_j)(y_{i}-y_j)=t(N+1) \sum_{i=1}^{N+1}x_{i}y_{i},$$
%
which implies the statement.
\end{proof}
Now we can find the forms  $G_{\mathcal{A}}^{(i)}, i=1,2.$
\begin{proposition}\label{identity for A_N}
The quadratic forms $G_{\mathcal{A}}^{(1)}, G_{\mathcal{A}}^{(2)}$ corresponding to the bilinear forms $G_{\mathcal{A}}^{(1)}(\cdot, \cdot)$, $G_{\mathcal{A}}^{(2)}(\cdot, \cdot)$ respectively have the following forms:
\begin{equation*}\label{G1 for A_N}
G_{\mathcal{A}}^{(1)}=(N+1)^{2}t^{2}\displaystyle \sum_{i , j=1}^{ N+1}(e^{i}\wedge e^{j})^{2},
\end{equation*}
\noindent and
\begin{equation}\label{G2 for A_N}
G_{\mathcal{A}}^{(2)}=t\displaystyle \sum_{i , j=1}^{ N+1}(e^{i}\wedge e^{j})^{2}.
\end{equation}
\end{proposition}
\begin{proof}
For the first equality we have
\begin{align*}
&G_{\mathcal{A}}^{(1)}
= t^2 \displaystyle \sum_{ i < j}^{ N+1}\sum_{ k < l}^{ N+1}\Big(  (e^i-e^j)\wedge (e^k-e^l)\Big)^{2}=\frac{t^2}{4}\displaystyle \sum_{ i , j,k,l=1}^{ N+1}\Big(  (e^i \wedge e^k) - (e^i \wedge e^l)-(e^j \wedge e^k)+(e^j \wedge e^l) \Big)^{2}\nonumber\\
&=t^2 (N+1)^2 \displaystyle \sum_{ i , j=1}^{ N+1}  (e^i \wedge e^j)^2
\end{align*}
since $\sum_{i=1}^{N+1}e^i|_{V}=0.$
For equality \eqref{G2 for A_N} we have by Lemma \ref{vectors and covectors in A_N} that
\begin{align}\label{sum 2 for G2}
 &G_{\mathcal{A}}^{(2)}=
 2t^{2} \displaystyle \sum_{1\leq i < j<k \leq N+1}\Bigg( (e^i-e^j)\big((e^i-e^k)^{\vee}\big)\big((e^i-e^j) \wedge (e^i-e^k)\big)^{2} \nonumber \\
 &+(e^i-e^j) \big((e^j-e^k)^{\vee} \big) \big((e^i-e^j)\wedge (e^j-e^k)\big)^{2}+(e^i-e^k) \big((e^j-e^k)^{\vee} \big) \big((e^i-e^k)\wedge (e^j-e^k)\big)^{2}\Bigg) \nonumber  \\
&=\frac{2t}{N+1}\displaystyle \sum_{1\leq i < j<k \leq N+1} (e^{i}\wedge e^{j}-e^{i}\wedge e^{k}+e^{j}\wedge e^{k})^{2}.
\end{align}
Note that
\begin{align}\label{sum over whole An for G2}
&\displaystyle \sum_{i,j,k=1}^{ N+1} (e^{i}\wedge e^{j}-e^{i}\wedge e^{k}+e^{j}\wedge e^{k})^{2}
=3(N+1)\displaystyle \sum_{ i , j=1}^{ N+1}  (e^i \wedge e^j)^2
\end{align}
since $\sum_{i=1}^{N+1}e^i|_{V}=0.$
Also it is easy to see that
\begin{equation}\label{relation between whole An and psitive.for G2}
 \displaystyle \sum_{1\leq i < j<k \leq N+1} (e^{i}\wedge e^{j}-e^{i}\wedge e^{k}+e^{j}\wedge e^{k})^{2}=\frac{1}{6} \displaystyle \sum_{i,j,k=1}^{ N+1} (e^{i}\wedge e^{j}-e^{i}\wedge e^{k}+e^{j}\wedge e^{k})^{2}.   
\end{equation}
Equality (\ref{G2 for A_N}) follows from formulas (\ref{sum 2 for G2})--(\ref{relation between whole An and psitive.for G2}).
\end{proof}
This leads us to the following result which can also be extracted from \cite{Martini 2003}.
\begin{theorem} \textup{(cf. \cite{Martini 2003})}  \label{lambda for general A_N}
Prepotential (\ref{F with extra variable y}), where $y=\sum_{i=1}^{N+1}x_i,$ for the configuration $(\mathcal{A},c)=(A_N^{+},t)$  satisfies WDVV equations 
\begin{equation*}
    F_{i} F_{j}^{-1} F_{k}=   F_{k} F_{j}^{-1} F_{i}, \quad i,j,k=1,\dots, N+1, 
\end{equation*}
where
$(F_i)_{p q}=\frac{\partial^{3}F}{\partial{x_i} \partial{x_p} \partial{x_q}},$
with
\begin{equation}\label{Lambda formula for general A_N}
\lambda=2(N+1)\sqrt{t}.
\end{equation}
\end{theorem}
\begin{proof}
 Firstly, $\mathcal{A}$ is a trigonometric $\vee$-system  by Proposition \ref{root system is trig. system}. 
Secondly, by Proposition~\ref{identity for A_N} we have that
\begin{align*}\label{L.H.S for A_n iden}
&G_{\mathcal{A}}^{(1)}-\frac{\lambda^2}{4}G_{\mathcal{A}}^{(2)}=\Big( (N+1)^{2}t-\frac{ \lambda^2}{4}   \Big) t \displaystyle \sum_{i , j=1}^{ N+1}(e^{i}\wedge e^{j})^{2},
\end{align*}
which is equal to $0$ for $\lambda$ given by (\ref{Lambda formula for general A_N}).
It follows by Theorem \ref{correction of Misha's Theorem} that $F$ satisfies WDVV equations (\ref{WDVV with y}) as a function on the hyperplane $V \subset \mathbb{C}^{N+1}$  which also depends on the auxiliary variable $y.$ 
Now we change variables to $(x_1 , \dots , x_{N+1})$ by putting $y=\sum_{i=1}^{N+1} x_i ,$ which implies the statement.
\end{proof}  

Let us now apply the restriction operation to the root system $A_N.$  
Let $n\in\mathbb{N}$ and ${\underline{m}=(m_1, \dots, m_{n+1})\in \mathbb{N}^{n+1}}$ be such that $\sum_{i=1}^{n+1} m_i = N+1$. 
Let us consider the subsystem $\mathcal{B}\subset \mathcal{A}$ given as follows:
\begin{equation*}
\mathcal{B}=\{e^{\sum_{j=1}^{i-1}m_{j}+k}-e^{\sum_{j=1}^{i-1}m_{j}+l},\quad 1\leq k<l\leq m_{i}, i=1,\dots,n+1\}.
\end{equation*}
The corresponding subspace $W_{\mathcal{B}}$ defined by (\ref{W_B the subspace}) can be given explicitly by the equations
\begin{align*}
\begin{dcases}
x_{1}=\dots =x_{m_{1}},\\
x_{m_{1}+1}=\dots =x_{m_{1}+m_{2}},\\
\vdots\\
x_{{ \sum_{i=1}^{n} m_{i}} + 1}=\dots =x_{N+1}.
\end{dcases}
\end{align*}
Define covectros $f^1,\dots,f^{n+1}\in W_{\mathcal{B}}^{\ast}$ by restrictions $f^{i}=\pi_{\mathcal{B}}\big( e^{\sum_{j=1}^{i}m_{j}} \big).$
The restriction of the configuration $A_{N}^{+}$ to the subspace $W_{\mathcal{B}}$ consists of the following covectors: 
\begin{align}\label{restricted system of A_N}
f^i-f^j, \quad  \text{with multiplicity} \quad t m_i m_j, \quad 1 \leq i<j \leq n+1.
\end{align}
The following result holds, which is closely related to a multi-parameter family of solutions found in \cite{Pavlov 2006} (see also \cite{Shen 2019}).
\begin{theorem}
The prepotential 
\begin{align}\label{restricted function for An}
 F(\xi)&=(\frac{1}{3}-t)y^3+ t y\sum_{k=1}^{n+1}m_{k}   \sum_{i=1}^{n+1} m_i  \xi_{i}^2  + 2 t^{3/2} \sum_{k=1}^{n+1}m_{k}  \sum_{i<j}^{n+1} m_{i}m_{j}f(\xi_i - \xi_j),
\end{align}
where $\xi=(\xi_1,\dots, \xi_{n+1})\in \mathbb{C}^{n+1}$ and $y=\sum_{i=1}^{n+1}\xi_i,$ 
satisfies WDVV equations 
\begin{equation*}\label{GDVV* 4 for restricted A_n}
 F_{i}F_{k}^{-1}F_{j}=F_{j}F_{k}^{-1}F_{i},\quad i,j,k=1,\dots,n+1,
\end{equation*}
where 
$(F_i)_{p q}=\frac{\partial^{3}F}{\partial{\xi_i} \partial{\xi_p} \partial{\xi_q}},$
for any generic $t,m_1,\dots,m_{n+1}\in \mathbb{C}.$
\end{theorem}
\begin{proof}
Let us suppose firstly that $m_i \in \mathbb{N}$ for all $i=1,\dots, n+1.$ Define $N=-1+\sum_{i=1}^{n+1}m_i.$
By Theorem \ref{lambda for general A_N} function (\ref{F with extra variable y}) with $\mathcal{A}=A_{N}^{+}$ and $\lambda$ given by (\ref{Lambda formula for general A_N}) is a solution of WDVV equations (\ref{WDVV with y}). 
By Theorem \ref{restricted system and WDVV} the prepotential given by
\begin{equation}\label{F with extra variable y for restricted A_N 2}
F(\xi,y)=\frac{1}{3}y^3+t y \sum_{i<j}^{n+1} m_i m_j (\xi_{i}-\xi_{j})^2+2(N+1)t^{3/2}\sum_{i<j}^{n+1} m_i m_j f(\xi_{i}-\xi_{j}), \quad \xi \in W_{\mathcal{B}},
\end{equation}
as a function on $W_{\mathcal{B}}\oplus \mathbb{C}$ satisfies WDVV equations.
Note that 
\begin{align}\label{general formula of quadratic part of F for A_n}
   \sum_{ i <j}^{ n+1}m_i m_j (\xi_i-\xi_j)^2 = \sum_{ i <j}^{ n+1}m_i m_j (\xi_{i}^{2}+\xi_{j}^{2})-2\sum_{ i <j}^{ n+1}m_i m_j \xi_{i} \xi_{j} .
\end{align}
Note also that
\begin{equation}\label{Quad. part 1 A_n}
  \sum_{ i <j}^{ n+1}m_i m_j (\xi_{i}^{2}+\xi_{j}^{2})=\frac{1}{2}\sum_{ i , j=1}^{n+1}m_i m_j (\xi_{i}^{2}+\xi_{j}^{2})-\sum_{i=1}^{n+1}m_{i}^{2} \xi_{i}^{2}=\sum_{i=1}^{n+1}(N+1-m_i)m_i \xi_{i}^{2},
\end{equation}
and that
\begin{equation}\label{Quad. part 2 A_n}
 \sum_{1\leq i <j\leq n+1}2 m_i m_j \xi_{i} \xi_{j} =\Big(\sum_{i=1}^{n+1}m_{i} \xi_{i} \Big)^{2}-\sum_{i=1}^{n+1}m_{i}^{2} \xi_{i}^{2}.  
\end{equation}
By making use (\ref{general formula of quadratic part of F for A_n})--(\ref{Quad. part 2 A_n}) the function (\ref{F with extra variable y for restricted A_N 2}) takes the form
\begin{equation}\label{sol. for restric, An}
    F(\xi,y)=\frac{1}{3}y^3+(N+1)t \sum_{i=1}^{n+1} m_i \xi_{i}^{2} y-t (\sum_{i=1}^{n+1}m_{i} \xi_{i})^{2}y+2(N+1)t^{3/2}\sum_{i<j}^{n+1} m_i m_j f(\xi_{i}-\xi_{j}).
\end{equation}
By setting $y=\sum_{i=1}^{n+1}m_i \xi_{i}$ and moving to variables $(\xi_1, \dots \xi_{n+1})\in \mathbb{C}^{n+1}$ solution (\ref{sol. for restric, An}) takes the required form (\ref{restricted function for An}).
The case of complex $m_i$ follows from the above since $F$ depends on $m_i$ polynomially. 
\end{proof}
\begin{remark}\label{another formula of y for restricted A_n}
We note that Theorem \ref{lambda for general A_N} and the solution $F$ given by (\ref{F with extra variable y}) is valid if one takes any generic linear combination of coordinates $x_i$ to form the extra variable  ${y=\sum_{i=1}^{N+1}a_i x_i, a_i\in \mathbb{C}}.$
The corresponding solution after restriction is given by the formula 
\begin{align*}
 F(\xi)&=\frac{1}{3}y^3+ty  \sum_{i<j}^{n+1} m_i m_j( \xi_{i}-\xi_{j})^{2} + 2(N+1) t^{3/2} \sum_{i<j}^{n+1} m_{i}m_{j}f(\xi_i - \xi_j),
\end{align*}
where $y$ is a linear combination of $\xi_1,\dots, \xi_{n+1}, \xi_{i}\in \mathbb{C}.$
\end{remark}
\section{ further examples in small dimensions}\label{section.4-dim trig.systems}
In section \ref{section.restriction of trig.systems} we presented the method of obtaining new solutions of WDVV equations through restrictions of known solutions. We applied it to classical families of root systems in sections \ref{section.BCn}, \ref{section.An}.
Similarly, starting from any root system and the corresponding solution of WDVV equations one can obtain further solutions by restrictions.  
In the next proposition we deal with a family of configurations in $4$-dimensional space which in general is not a restriction of a root system. 

\begin{proposition}\label{4-dim trig. sys.}
Let a configuration $\mathcal{A}\subset \mathbb{C}^{4}$ consist of the following covectors:
\begin{align*}
e^i, \quad  \text{with multiplicity} \quad p, \quad 1 \leq i \leq 3,\\
 e^4, \quad \text{with multiplicity} \quad q,\\
 e^i \pm e^j, \quad \text{with multiplicity} \quad r, \quad 1\leq i < j \leq 3,\\
  \frac{1}{2}(e^1 \pm e^2\pm e^3 \pm e^4), \quad \text{with multiplicity} \quad s,
\end{align*}
where $p,q,r,s \in \mathbb{C}$ are such that $4r+s \neq 0.$
Then $\mathcal{A}$ 
is a trigonometric $\vee$-system if
\begin{align}
p&=2 r+ s,\label{con.1}\\
q&= \frac{s(s-2r)}{4r+s}\label{con.2},
\end{align}
and $ps \neq 0.$
The corresponding prepotential (\ref{F with extra variable y}) with 
\begin{equation}\label{lambda for 4-dim}
 \lambda=6 \sqrt{3} (2r+s)(4r+s)^{-1/2}  
\end{equation}
is a solution of WDVV equations.
\end{proposition}
\begin{proof}
For $x=(x_1,x_2,x_3,x_4), y=(y_1,y_2,y_3,y_4)\in \mathbb{C}^4$ the bilinear form  $G_{\mathcal{A}}$ is given by
$$G_{\mathcal{A}}(x,y)=(p+4 r+2s)(x_1 y_1+x_2 y_2+x_3 y_3)+(q+2s)x_4 y_4.$$
To simplify notations let us introduce covectors
\begin{align*}
\alpha_1&={\frac{1}{2}}(e^1+e^2+e^3+e^4), \quad \alpha_2={\frac{1}{2}}(e^1+e^2+e^3-e^4), \\
\alpha_3&={\frac{1}{2}}(e^1+e^2-e^3+e^4), \quad \alpha_4={\frac{1}{2}}(e^1-e^2+e^3+e^4), \\
\alpha_5&={\frac{1}{2}}(e^1-e^2-e^3+e^4), \quad \alpha_6={\frac{1}{2}}(e^1-e^2+e^3-e^4), \\
\alpha_7&={\frac{1}{2}}(e^1+e^2-e^3-e^4), \quad \alpha_8={\frac{1}{2}}(e^1-e^2-e^3-e^4). 
\end{align*}
Because of $B_3 \times A_1$-symmetry 
it is enough to check $\vee$-conditions for the following series only:
\begin{align*}
\Gamma_{e_1}&=\{\alpha_1,\alpha_8 \},\quad
\Gamma_{e_4}=\{\alpha_1,\alpha_2 \},\quad
\Gamma_{e^1+e^2}=\{\alpha_1,\alpha_7  \}, \quad
\Gamma_{\alpha_1}^{1}=\{\alpha_2,e^4\},\quad 
\Gamma_{\alpha_1}^{2}=\{\alpha_3,e^3\},\\
\Gamma_{\alpha_1}^{3}&=\{\alpha_5,e^2+e^3\}.
\end{align*}
Trigonometric $\vee$-system conditions for the series 
$\Gamma_{e_1},\Gamma_{e_4},\Gamma_{e^1+e^2}$ are immediate to check.
Let us consider the trigonometric $\vee$-condition for ${\alpha_1}$-series.
We have 
\begin{align*}
 \alpha_1(\alpha_{2}^{\vee})&=\frac{3q+4s-p-4r}{4(p+4r+2s)(q+2s)},\quad 
 \alpha_1(e^{4\vee})=\frac{1}{2(q+2s)},\quad
\alpha_1 \wedge \alpha_2 =- \alpha_1 \wedge e^4,
\end{align*}
which implies $\vee$-condition for $\Gamma_{\alpha_{1}}^{1}$ since 
$s(3q+4s-p-4r)-2q(p+4r+2s)=0$
by relations (\ref{con.1}), (\ref{con.2}).

Also we have 
\begin{align*}
 \alpha_1(\alpha_{3}^{\vee})&=\frac{q+p+4r+4s}{4(p+4r+2s)(q+2s)},\quad 
 \alpha_1(e^{3\vee})=\frac{1}{2(p+4r+2s)},\quad
\alpha_1 \wedge \alpha_3 =- \alpha_1 \wedge e^3,
\end{align*}
which implies $\vee$-condition for $\Gamma_{\alpha_1}^{2}$ since
$s(q+p+4r+4s)-2p(q+2s)=0$
by relations (\ref{con.1}), (\ref{con.2}).

Finally,
we have 
\begin{align*}
 \alpha_1(\alpha_{5}^{\vee})&=\frac{p+4r-q}{4(p+4r+2s)(q+2s)},\, 
 \alpha_1((e^2+e^3)^{\vee}) =\frac{1}{p+4r+2s},\,
\alpha_1 \wedge \alpha_5 =- \alpha_1 \wedge(e^2+e^3),
\end{align*}
which implies $\vee$-conditions for $\Gamma_{\alpha_1}^{3}$ since
$s(p+4r-q)-4r(q+2s)=0$
by relations (\ref{con.1}), (\ref{con.2}).

Let us now find the quadratic form $G_{\mathcal{A}}^{(1)}.$
By straightforward calculations we get
\begin{align}\label{R.H.S for 4-dim}
 & G_{\mathcal{A}}^{(1)}=2\big(p^2+8pr+4ps+16rs+16r^2+4s^2\big) \sum_{i<j}^{3}\big( e^i \wedge e^j\big)^{2}\nonumber \\
 &+2(pq+2ps+4qr+2qs+8rs+4s^2) \sum_{i=1}^{3}\big( e^i \wedge e^4\big)^{2}\nonumber \\
 &=\frac{18(2r+s)^2}{4r+s}\Bigg( (4r+s)\sum_{i<j}^{3}\big( e^i \wedge e^j\big)^{2}+s \sum_{i=1}^{3}\big( e^i \wedge e^4\big)^{2}\Bigg).
\end{align}
Now let us find the quadratic form $G_{\mathcal{A}}^{(2)}.$
We have 
\begin{align}\label{G2 for the 4-dim}
    & G_{\mathcal{A}}^{(2)}=pr \sum_{i=1}^{3}\sum_{j< k}^{3} e^{i}\big((e^j \pm e^k)^{\vee}\big)\big(e^i\wedge (e^j\pm e^k)\big)^{2} +ps \sum_{i=1}^{3}\sum_{j=1}^{8} e^{i}(\alpha_{j}^{\vee})\big(e^i\wedge \alpha_{j}\big)^{2}\nonumber \\
    &+rs\sum_{i<j}^{3}\sum_{k=1}^{8}(e^i\pm e^j)(\alpha_{k}^{\vee})\big( (e^i\pm e^j) \wedge \alpha_{k} \big)^{2} + r^2 \sum_{i<j}^{3}\sum_{k<l}^{3} (e^i\pm e^j)((e^k\pm e^l)^{\vee})\big( (e^i\pm e^j) \wedge(e^k\pm e^l)  \big)^2 \nonumber \\
    &+s^2 \sum_{i,j=1}^{8}\alpha_{i}(\alpha_{j}^{\vee}) (\alpha_i \wedge \alpha_j)^2 =\frac{4pr}{p+4r+2s}\sum_{i< j}^{3}(e^i\wedge e^j)^{2} + \frac{2ps}{p+4r+2s}\sum_{i=2}^{4}(e^1\wedge e^i)^{2}\nonumber \\
    &+\frac{4rs}{p+4r+2s}\big((e^1 \wedge e^2)^{2}+(e^1\wedge e^3)^{2}+2(e^2 \wedge e^3)^{2}+2(e^1 \wedge e^4)^{2}+(e^2 \wedge e^4)^{2}+(e^3 \wedge e^4)^{2} \big) \nonumber \\
    &+\frac{8r^{2}}{p+4r+2s}\sum_{i< j}^{3}(e^i\wedge e^j)^{2}+\frac{2 s^{2}}{p+4r+2s}\Big((e^2 \wedge e^3)^{2}- \frac{4r}{s} (e^1\wedge e^4)^{2}+(e^2 \wedge e^4)^{2}+(e^3 \wedge e^4)^{2} \Big).
\end{align}
By making further use of relations (\ref{con.1}), (\ref{con.2}) the expression (\ref{G2 for the 4-dim}) can be simplified to the form
\begin{align}\label{L.H.S for 4-dim}
 & G_{\mathcal{A}}^{(2)}=\frac{2}{3} (4r+s)\sum_{i< j}^{3}(e^i\wedge e^j)^{2}+ \frac{2}{3} s \sum_{i= 1}^{3}(e^i\wedge e^4)^{2}.
\end{align}
The final statement of the proposition follows from formulas (\ref{R.H.S for 4-dim}), (\ref{L.H.S for 4-dim}) and Theorem \ref{correction of Misha's Theorem}.
\end{proof}
\begin{remark}\label{special choices of multip. for 4-dim example}
We note that for special values of the parameters configuration $\mathcal{A}$ is a restriction of a root system (cf. \cite{Misha&Veselov 2007} where the rational version of this configuration was considered).
Thus if $r=0$ and $p=q=s$ then $\mathcal{A}$ reduces to the root system $D_{4}$. 
If $r=1, s=4,$ then $p=6$ and $q=1$ and the resulting configuration is the restriction of the root system $E_7$ along subsystem of type $A_3.$ 
If $s=2r$ then the resulting configuration is the restriction of the root  system $E_6$ along subsystem of type  $A_1 \times A_1.$ 
\end{remark}
Further solutions of WDVV equations can be obtained from Proposition \ref{4-dim trig. sys.} by restricting the configuration $\mathcal{A}.$ 
\begin{proposition}\label{Restriction of 4-dim}
Let $\mathcal{A}_1 \subset \mathbb{C}^3$ be the configuration 
%
$$\mathcal{A}_1=\{2e^1,e^1, e^2,e^3, e^1 \pm e^2,\frac{1}{2}(e^2 \pm e^3), \frac{1}{2}(2e^1 \pm e^2\pm e^3)\},$$    
%
with the corresponding multiplicities $\{r,2p, p, q , 2r, 2s, s \},$ where $p,q,r,s \in \mathbb{C}.$
Let configuration $\mathcal{A}_2\subset \mathbb{C}^3$ consist of the following set of covectors:
\begin{align*}
e^i, \quad  \text{with multiplicity} \quad p+s, \quad 1 \leq i \leq 3,\\
 e^i + e^j, \quad \text{with multiplicity} \quad r+s, \quad 1\leq i < j \leq 3,\\
  e^i - e^j, \quad \text{with multiplicity} \quad r, \quad 1\leq i < j \leq 3,\\
 e^1 + e^2 + e^3, \quad \text{with multiplicity} \quad q+s.
\end{align*}
Suppose that relations (\ref{con.1}), (\ref{con.2}) hold and that $ps(4r+s) \neq 0.$  
Then $\mathcal{A}_1, \mathcal{A}_2$ are trigonometric $\vee$-systems which also define solutions of WDVV equations given by formula (\ref{F with extra variable y}) with $\lambda$ given by (\ref{lambda for 4-dim}).
\end{proposition}
Proof of this proposition follows from an observation that configuration $\mathcal{A}_1$ can be obtained from the configuration $\mathcal{A}$ from Proposition \ref{4-dim trig. sys.} by restricting it to the hyperplane $x_1=x_2$ (up to renaming the vectors). 
Similarly, configuration $\mathcal{A}_2$ can be obtained by restricting the configuration $\mathcal{A}$ to the hyperplane $x_1+x_2+x_3-x_4=0$ (and up to renaming the vectors). 
Other three-dimensional restrictions of the configuration $\mathcal{A}$ give restriction of the root system  $F_4$ and a configuration from $BC_3$ family.

Rational versions of configurations $\mathcal{A}_1 , \mathcal{A}_2$ were considered in \cite{Misha&Veselov 2007}. Note that configuration  $\mathcal{A}_1$ has collinear vectors $2e^1, e^1,$ so its rational version has different size.

Two-dimensional restrictions of  $\mathcal{A}$ are considered below in Proposition \ref{6-vectors on the plane} and Proposition~\ref{8-vectors trig.v-sys}, or can belong to $BC_2$ family of configuration, or have the form of configuration $G_2$ or appear in \cite{Misha2009}*{Proposition 5}.

Let us now consider examples of solutions (\ref{F with extra variable y}) of WDVV equations where configuration $\mathcal{A}$ contains a small number of vectors on the plane.
The next two propositions confirm that trigonometric $\vee$-systems with up to five covectors belong to $A_2$ or $BC_2$ families. 
\begin{proposition}\label{3 vectors in An}
Any irreducible trigonometric $\vee$-system  $\mathcal{A}\subset \mathbb{C}^2$ consisting of three vectors with non-zero multiplicities has the form (\ref{restricted system of A_N}) where $n=2$ for some values of parameters. 
\end{proposition}
\begin{proof}
By \cite{Misha2009}*{Proposition 2}
any such configuration has the form $\mathcal{A}=\{\alpha,\beta,\gamma \}$ with the corresponding multiplicities $\{c_{\alpha}, c_{\beta}, c_{\gamma}\}$, where vectors in $\mathcal{A}$ satisfy $\alpha \pm \beta \pm \gamma =0$ for some choice of signs.
It is easy to see that equations
$$ tm_{1}m_{2}=c_{\alpha},\quad 
   tm_{1}m_{3}=c_{\beta},\quad
   tm_{2}m_{3}=c_{\gamma},$$
for $m_1, m_2,m_3, t \in \mathbb{C}$ can be resolved.
\end{proof}
\begin{proposition}\label{4 and 5 vectors in BCn}
Any irreducible trigonometric $\vee$-system  $\mathcal{A}\subset \mathbb{C}^2$ consisting of four or five vectors with non-zero multiplicities has the form $BC_2(r,s,q;\underline{m})$ for some values of parameters. 
\end{proposition}
\begin{proof}
By \cite{Misha2009}*{Proposition 3} any irreducible trigonometric $\vee$-system  $\mathcal{A}$ consisting of four vectors has the form
$\mathcal{A}=\{2e^{1},2e^{2},e^{1}\pm e^{2} \}$ in a suitable basis, and the corresponding multiplicities $\{c_1, c_2,c_0\}$ where $c_0 \neq -2c_i$ for $ i=1,2.$
Now we require parameters $r,s,q,m_1,m_2$ to satisfy
\begin{align*}
    &sm_1 +\frac{1}{2}qm_{1}(m_{1}-1)=c_{1},\\
      &sm_2 +\frac{1}{2}qm_{2}(m_{2}-1)=c_{2},\\
     & qm_{1}m_{2}=c_{0}, \quad r=0,
\end{align*}
which can be done by taking
\begin{align*}
    s&=\frac{1}{m_1}\big(c_{1}-\frac{c_{0}(m_{1}-1)(2c_{1}+c_0)}{2m_{1}(2c_{2}+c_{0})}    \big),\\
   q&=\frac{c_0(2c_{1}+c_0)}{m_{1}^{2}(2c_{2}+c_0)}, \quad m_{2}=\frac{(2c_{2}+c_{0})m_{1}}{2c_{1}+c_{0}},\quad m_1 \in \mathbb{C}\setminus \{0 \}.
\end{align*}

By \cite{Misha2009}*{Proposition 4}
any irreducible trigonometric $\vee$-system  $\mathcal{B}$ consisting of five covectors in a suitable basis has the form $\mathcal{B}=\{e^{1},2e^{1},e^{2},e^{1}\pm e^{2} \},$ and the corresponding multiplicities  $\{c_1,\widetilde{c_{1}}, c_2, c_{\pm}\}$ satisfy $c_{+} = c_{-}$ and
$2 \widetilde{c_{1}}c_{2}=c_{+}(c_{1}-c_{2}),$
where $(c_1 +4 \widetilde{c_{1}} +2c_{+})(c_2 +2c_{+})\neq 0.$
In order to compare the configuration $\mathcal{B}$ with the configuration $BC_2(r,s,q;\underline{m})$, we require parameters $r,s,q,m_1,m_2$ to satisfy
\begin{align*}
    & rm_{1}=c_1, \quad   rm_{2}=c_2,\quad  qm_{1}m_{2}=c_{+},\\
    &sm_1 +\frac{1}{2}qm_{1}(m_{1}-1)=\widetilde{c}_{1},\quad sm_2 +\frac{1}{2}qm_{2}(m_{2}-1)=0.
\end{align*}
These equations can be solved by taking
\begin{align*}
    r=\frac{c_{1}}{m_{1}},\quad  s=\frac{c_{+}(c_1-c_2 m_1)}{2c_{2}m_{1}^{2}}\quad q=\frac{c_{+}c_1}{c_{2}m_{1}^{2}}, \quad m_{2}=\frac{c_{2}m_{1}}{c_{1}}, \quad m_1 \in \mathbb{C}\setminus \{ 0 \}.
\end{align*}
\vspace{-3mm}
\end{proof}

In the rest of this section we give more examples of trigonometric $\vee$-systems on the plane, which can be checked directly or using Theorem \ref{correction of Misha's Theorem}.
The configuration in the following proposition can be obtained by restricting configuration $\mathcal{A}_1$ from Proposition \ref{Restriction of 4-dim} to the plane $2x_1+x_2-x_{3}=0.$
\begin{proposition}\label{6-vectors on the plane}
Let $\mathcal{A}=\{e^{1},2e^{1},e^{2},{e^{1}+e^{2}},{e^{1}-e^{2}},{2e^{1}+e^{2}}\}\subset \mathbb{C}^2$ with the corresponding multiplicities $\{4a,a,2a,2a,2(a-b),\frac{2ab}{4a-3b}\}$, where ${4a-3b}\neq 0.$ Then $\mathcal{A}$ is a trigonometric
$\vee$-system provided that $a(2a-b)\neq 0.$ The corresponding solution of the WDVV equations has the form (\ref{F with extra variable y}) with 
$\lambda=6\sqrt{3}(2a-b)(4a-3b)^{-1/2}.$ 
\end{proposition}
The configuration in the following proposition can be obtained by restricting configuration $\mathcal{A}_1$ from Proposition \ref{Restriction of 4-dim} to the plane $x_{3}=0.$
\begin{proposition}\label{8-vectors trig.v-sys}
Let $\mathcal{A}=\{e^{1},2e^{1},e^{2},2e^{2},e^{1} \pm e^{2},e^{1}\pm 2e^{2}\} \subset \mathbb{C}^2$ with the corresponding multiplicities $\{2a, \frac{a}{2}-\frac{b}{4},2b,a,b,a-\frac{b}{2} \}$, where $a\neq 0.$ Then $\mathcal{A}$ is a trigonometric $\vee$-system
and the corresponding solution of the WDVV equations has the form (\ref{F with extra variable y}) with 
${\lambda=6\sqrt{6}a\big(4a-b\big)^{-1/2}}.$ 
\end{proposition}
In the next two propositions we give examples of trigonometric $\vee$-systems with nine and ten covectors on the plane.
\begin{proposition} \textup{(cf. \cite{Misha2009})}  
\label{9-vectors trig.v-sys}
Let $\mathcal{A}=\{e^{1},2e^{1},e^{2},e^{1} \pm e^{2},\frac{1}{2}(3e^{1} \pm e^{2}),\frac{1}{2}(e^{1} \pm e^{2})\}\subset \mathbb{C}^2$ with the corresponding multiplicities $\{a,b,\frac{a}{3},b,\frac{a}{3},a\}.$
Then  $\mathcal{A}$ is a trigonometric $\vee$-system provided that $a\neq -2b.$
The corresponding solution of the WDVV equations has the form (\ref{F with extra variable y}) with 
$\lambda=6(a+2b)(a+4b)^{-1/2}.$ 
\end{proposition}

Note that if $b=0$ then after rescaling $e^2 \to \sqrt{2} e^2$ this configuration reduces to the positive half of the root system $G_2$.

\begin{proposition} \textup{(cf. \cite{Misha2009})} 
\label{10-vectors trig.v-sys}
Let $\mathcal{A}=\{e^{1},2e^{1},e^{2},2e^{2},e^{1} \pm e^{2},e^{1} \pm 2e^{2},2e^{1} \pm e^{2}\} \subset \mathbb{C}^2$ with the corresponding multiplicities $\{6a, \frac{3a}{2},6a, \frac{3a}{2},4a,a,a  \}.$
Then  $\mathcal{A}$ is a trigonometric $\vee$-system provided that $a\neq 0.$ 
The corresponding solution of the WDVV equations has the form (\ref{F with extra variable y}) with 
$\lambda=15 a^{1/2}.$ 
\end{proposition}
\section{root systems solutions revisited}\label{section.root systems solutions revisited}
Following \cites{Martini 2003 (1), Martini 2003}, recall that WDVV equations (\ref{WDVV with y}) have solutions of the form
\begin{equation}\label{F (Mrtini).V2}
\widetilde{F}(x,y)=\frac{\gamma}{6}y^3+\frac{\gamma}{2}y \langle x,x \rangle +\sum_{\alpha\in\mathcal{R}^{+}}c_{\alpha}\widetilde{f}(\alpha(x)),\quad (x\in V, y\in \mathbb{C}),
\end{equation}
where $\mathcal{R}\subset V^{\ast}$ is a root system of rank $N$, multiplicities $c_{\alpha}$ and the inner product $\langle\cdot,\cdot \rangle$ are invariant under the Weyl group, $\gamma=\gamma_{(\mathcal{R},c)} \in \mathbb{C}$ and function $\widetilde{f}$ is given by (\ref{f for martini}). 
The corresponding values of $\gamma_{(\mathcal{R},c)}$ were given explicitly in \cites{Martini 2003 (1), Martini 2003} for constant multiplicity functions ${c_{\alpha}=t \,\, \forall \alpha}$ (except for $\mathcal{R}= BC_N, G_2$),
they were found in \cite{Bryan 2008} for special multiplicities and in \cites{Shen 2018, Shen 2019} for arbitrary (non-reduced) root system $R$ with invariant multiplicity.
For type $E$ root systems we have
$$\gamma_{(E_6,t)}=2i\sqrt{6} t, \quad
\gamma_{(E_7,t)}=4i\sqrt{3} t, \quad
\gamma_{(E_8,t)}= 2 i\sqrt{30} t.$$
Similarly to analysis of the $BC_N$ case in Section \ref{section.BCn} these solutions lead to solutions $F$ of the form (\ref{F with extra variable y}) for $\mathcal{A}=\mathcal{R}^{+}$ and the corresponding values of $\lambda=\lambda_{(\mathcal{R},c)}$ are given by 
\begin{equation}\label{lambda for type E678}
\lambda_{(E_6,t)}= 12\sqrt{2t}, \quad
\lambda_{(E_7,t)}=9\sqrt{6t}, \quad
\lambda_{(E_8,t)}= 30\sqrt{t}.
\end{equation}

We recall that $\lambda_{(\mathcal{R},c)}$, in contrast to $\gamma_{(\mathcal{R},c)}$, is invariant under linear transformations applied to $\mathcal{R}.$
An alternative way to derive   values  (\ref{lambda for type E678}) is to apply Theorem \ref{restricted system and WDVV} to already known solutions. Thus $\lambda_{(E_6,t)}$ can be derived, for example, by considering the four-dimensional restriction of ${E_6}$ along a subsystem of type $A_1 \times A_1$ as this restriction is equivalent to the configuration from Proposition \ref{4-dim trig. sys.} when parameter $s=2r.$
Likewise restriction of ${E_7}$ along a subsystem of type $A_3$ gives the same configuration from Proposition \ref{4-dim trig. sys.} with $r=1$ and $s=4.$  
Similarly, restriction of ${E_8}$ along a subsystem of type $D_6$ gives the configuration of type $BC_2$ which allows to get $\lambda_{(E_8,t)}.$ 

Let us now find $\lambda_{(\mathcal{R},c)}$ for the remaining cases, namely, $\mathcal{R}=F_4$ and $\mathcal{R}=G_2,$ and general multiplicity.
We start with the root system $\mathcal{R}= F_4.$

\begin{proposition}\label{Prop.lambda for F4}
Let $\mathcal{A}= F_{4}^{+}$ be the positive half of the root system $F_4$ with the multiplicity function $c$ given by
\begin{align}\label{The system F4}
c \big(  \frac{1}{2}(e^1\pm e^2 \pm e^3 \pm e^4)\big)= c(e^i)= s, \quad (1\leq i \leq 4), \nonumber \\
 c(e^i \pm e^j)=r,  \quad (1\leq i < j \leq 4),
\end{align}
where $r,s \in \mathbb{C}.$
Then in the corresponding solution (\ref{F with extra variable y}) of the WDVV equations (\ref{WDVV with y}) we have
\begin{equation}\label{lambda for F4}
    \lambda= \lambda_{(F_4,c)}={6\sqrt{3}(2r+s)}(4r+s)^{-1/2}.
\end{equation}
\end{proposition}
\begin{proof}
We note that the restriction of the configuration defined in Proposition \ref{4-dim trig. sys.} to the hyperplane $x_4=0$ gives the same configuration as one gets by restricting $\mathcal{A=}F_4^{+}$ to the hyperplane $x_4=0.$ Hence $\lambda$ is given by formula (\ref{lambda for 4-dim}).
\end{proof}
Proposition \ref{Prop.lambda for F4} has the following implication for the corresponding solution of the form (\ref{F (Mrtini).V2}), which is also contained in \cite{Shen 2019}.
\begin{proposition}  \cite{Shen 2019} \label{F4 Martini's solution}
For $\mathcal{R}=F_4$ with the multiplicity function (\ref{The system F4}) we have 
$$
{\gamma_{(F_4,c)}^2=-(s+2r)(s+4r).}    $$
\end{proposition}
\begin{proof}
We have 
$\sum_{\alpha\in F_{4}^{+}}c_{\alpha} \alpha(x)^2=3(s+2r)\sum_{i=1}^{4}  x_{i}^{2}.$ 
 Then solution $F$ given by (\ref{F with extra variable y}) for $\mathcal{A}=F_4^{+}$ takes the form 
\begin{equation}\label{F_1 for F_4}
F(\widetilde{x},\widetilde{y})=\frac{1}{3}\widetilde{y}^3+3(s+2r)\widetilde{y} \sum_{i=1}^{4}\widetilde{x}_{i}^{2}+\lambda\sum_{\alpha\in F_4^{+}}c_{\alpha}f(\alpha(\widetilde{x})),
\end{equation}
where $\lambda$ is given by \eqref{lambda for F4}, and we redenoted variables $(x,y)$ by $(\widetilde{x},\widetilde{y}).$
By dividing by $-\lambda$ and changing variables $\widetilde{x}=-ix$,  $\widetilde{y}=\frac{\gamma \lambda}{6(s+2r)}y,$ solution (\ref{F_1 for F_4}) takes the form (\ref{F (Mrtini).V2}) provided that $\gamma^2 \lambda^2 =-108(s+2r)^2,$ which implies the statement.
\end{proof}
Let us now find the value of $\lambda$ for $\mathcal{R}=G_2.$
%
%
\begin{proposition}\label{Prop.lambda for G2. plane}
Let $\mathcal{A}=G_{2}^{+}$ be the positive half of the root system $G_2$ with the multiplicity function given by
\begin{equation}\label{Symmetric representation of on the plane}
c(\sqrt{3}e^1)=c( \frac{\sqrt{3}e^1}{2}\pm \frac{3e^2}{2})=q, \quad c(e^2)=c( \frac{\sqrt{3}e^1}{2}\pm \frac{e^2}{2})=p,
\end{equation}
where $q,p \in \mathbb{C}.$
Then in the corresponding solution (\ref{F with extra variable y}) of the WDVV equations (\ref{WDVV with y}) we have
\begin{equation}\label{lambda for G2.plane represntation}
    \lambda=\lambda_{(G_2,c)}={6(p+3q)}(p+9q)^{-1/2}.
\end{equation}
\end{proposition}
\begin{proof}
Note that by restricting the configuration $\mathcal{A}_2$ defined in Proposition \ref{Restriction of 4-dim} to the hyperplane $x_1 +x_2 +x_3=0$ we get the two-dimensional configuration
$$\widetilde{\mathcal{A}}_2=\{e^1,e^2,e^1 + e^2 ,e^1 - e^2,e^1 + 2e^2,2e^1 + e^2 \}$$
which can be mapped to the configuration $G_2$ by a linear transformation.
The corresponding multiplicities satisfy
 $$   p=3(r+s), \quad q=r,$$
which implies the statement by Proposition \ref{Restriction of 4-dim} and Theorem \ref{restricted system and WDVV}.
\end{proof}
Proposition \ref{Prop.lambda for G2. plane} has the following implication for the corresponding solution of the form (\ref{F (Mrtini).V2}), which is also contained in \cite{Shen 2019}.
\begin{proposition}\cite{Shen 2019} \label{G2 Martini's solution on the plane}
For $\mathcal{R}=G_2$ with multiplicity function (\ref{Symmetric representation of on the plane})
we have 
$$\gamma_{(G_2,c)}^2=-\frac{3}{8}(p+3q)(p+9q).$$
\end{proposition}
\begin{proof}
We have 
$\sum_{\alpha\in G_{2}^{+}}c_{\alpha} \alpha(x)^2=\frac{3}{2}(p+3q)\sum_{i=1}^{2}  x_{i}^{2}.$ 
Then solution $F$ given by (\ref{F with extra variable y}) for $\mathcal{A}=G_{2}^{+}$ takes the form 
\begin{equation}\label{F_1 for G_2 .plane}
F(\widetilde{x},\widetilde{y})=\frac{1}{3}\widetilde{y}^3+\frac{3}{2}(p+3q)\sum_{i=1}^{2}  x_{i}^{2}\widetilde{y} +\lambda\sum_{\alpha\in G_2^{+}}c_{\alpha}f(\alpha(\widetilde{x})),
\end{equation}
 where $\lambda$ is given by \eqref{lambda for G2.plane represntation}, and we redenoted variables $(x,y)$ by $(\widetilde{x},\widetilde{y}).$
By dividing by $-\lambda$ and changing variables $\widetilde{x}=-ix$,  $\widetilde{y}=\frac{\gamma \lambda}{3(p+3q)}y,$ solution (\ref{F_1 for G_2 .plane}) takes the form (\ref{F (Mrtini).V2}) provided that $\gamma^2 \lambda^2 =\frac{27}{2}(p+3q)^3$ which implies the statement.
\end{proof}
Solutions of WDVV equations of the form (\ref{F (Mrtini).V2}) were also obtained in \cite{Bryan 2008}.
More exactly, consider the multiplication $\ast$ on the tangent space $T_{(x,y)}(V \oplus U)\cong V \oplus U,$  
where ${\dim U=1, x\in V, y\in U}$ which is given by 
\begin{equation}\label{algebra from Bryan 2}
   u\ast v=\langle u,v \rangle E + \widetilde{\gamma}^{-1} \sum_{\beta\in \mathcal{R}^{+}} \frac{c_{\beta}}{\langle \beta,\beta \rangle}  \beta (u) \beta (v) \coth \beta(x) \beta, \quad (u, v \in V),
\end{equation}
 $\widetilde{\gamma}=\widetilde{\gamma}_{(\mathcal{R},c)} \in \mathbb{C}$ and $E\in U$ is the identity of $\ast $.
It was shown in \cite{Bryan 2008} that this multiplication is associative.
It can be seen (cf. section \ref{section.Trig.sys and WDVV} above) that associativity of \eqref{algebra from Bryan 2} is equivalent to the statement that function
\begin{equation}\label{F for Bryan.v1}
\widetilde{F}(x,y)=\frac{\widetilde{\gamma}}{6}y^3+ \frac{\widetilde{\gamma}}{2}y\langle x,x \rangle    + \sum_{\alpha\in \mathcal{R}^{+}}{d_{\alpha}} \widetilde{f}(\alpha(x)),
\end{equation}
where $d_{\alpha}=\frac{c_{\alpha}}{\langle \alpha,\alpha \rangle}$
satisfies WDVV equations, hence 
$\widetilde{\gamma}=\widetilde{\gamma}_{(\mathcal{R},c)}=\gamma_{(\mathcal{R},d)}.$

Let $\{\alpha_{1}, \dots, \alpha_{N} \}$ be a basis of simple roots of $\mathcal{R}$. Recall that there exists the \textit{highest root}
$\theta=\theta_{\mathcal{R}}=\sum_{i=1}^{N} n_i \alpha_i \in \mathcal{R}$
such that, for every $\beta=\sum_{i=1}^{N} p_i \alpha_i \in \mathcal{R},$ we have $n_i \geq p_i$ for all $i=1, \dots, N$ 
\cite{Bourbaki}.
The constant $\widetilde{\gamma}=\widetilde{\gamma}_{(\mathcal{R},c)}$ was expressed in \cite{Bryan 2008} in terms of the highest root of the root system $\mathcal{R}.$
\begin{proposition}\label{Epsilon.1}\cite{Bryan 2008}
The value of $\widetilde{\gamma}_{(\mathcal{R},c)}$ in the solution (\ref{F for Bryan.v1}) in the case of constant multiplicity function $c_{\alpha}=t$ is given by 
%
 $$  \widetilde{\gamma}_{(\mathcal{R},c)}^2=-\frac{t^2}{8}\Big( \langle \theta,\theta\rangle + \sum_{i=1}^{N}n_{i}^{2}\langle \alpha_i,\alpha_i \rangle\Big).$$
\end{proposition}
Now we give a generalization of Proposition \ref{Epsilon.1} to the case of non-constant multiplicity function. Let $p$ be the multiplicity of short roots and $q$ be the multiplicity of long roots in a reduced non-simply laced root system $\mathcal{R}.$ 
\begin{proposition}\label{Prop. values of epsilon}
We have 
\begin{equation}\label{General epsilon for non-simply laced systems}
   \widetilde{\gamma}_{(\mathcal{R},c)}^2=-\frac{1}{8}\Big( a_{0} \langle \theta,\theta \rangle +\sum_{i=1}^{N}  a_{i} n_{i}^{2} \langle \alpha_i,\alpha_i \rangle \Big),
\end{equation}
where scalars $ a_{i}$ for all irreducible reduced non-simply laced root systems are given as follows.
(1) Let ${\mathcal{R}}=B_N$
with the basis of simple roots
${\alpha_1={e^1-e^2}, \dots, \alpha_{N-1}=e^{N-1}-e^N, \alpha_N=e^N}.$
Then
 \begin{align}\label{ai of Bn}
a_{0}= a_1= a_N= p q, \quad
a_{i} = q^2,  \quad (2 \leq i  \leq N-1).
\end{align}
(2) Let ${\mathcal{R}}=C_N$
with the basis of simple roots
${\alpha_1={e^1-e^2}, \dots, \alpha_{N-1}=e^{N-1}-e^N, \alpha_N=2e^N.}$
 Then 
 \begin{align}\label{ai of Cn}
a_{0}= a_1= a_N= p q, \quad
a_{i} = p^2,  \quad (2 \leq i  \leq N-1). 
\end{align}   
(3) Let ${\mathcal{R}}=F_4$
with the basis of simple roots
${\alpha_1={e^2-e^3}, \alpha_2={e^3-e^4}, \alpha_3={e^4}}$,

${\alpha_4=\frac{1}{2}(e^1-e^2-e^3-e^4)}.$
Then 
 \begin{align}\label{ai of F4}
a_{0}=a_{2}=a_{4}= p q, \quad
a_{1} =p^2, \quad a_{3}=q^2. 
\end{align}
(4)
Let ${\mathcal{R}}=G_2$
with the basis of simple roots
$\alpha_1={\frac{\sqrt{3}e^1}{2}}-\frac{3e^2}{2}, \alpha_2={e^2}.$ Then
 \begin{align}\label{ai of G2}
a_{0}= p^2, \quad
a_{1} =pq, \quad a_{2}=q^2. 
\end{align}
\end{proposition}
\begin{proof}
It follows from Proposition \ref{solution for BC_n of Martini theorem} that 
$$\widetilde{\gamma}_{(B_N,c)}^2=-q(p+(N-2)q).$$
Note that 
$\theta_{B_N} =e^1+e^2=\alpha_1+2(\alpha_2+\dots +\alpha_{N}).$
Then it is easy to see that the substitution of (\ref{ai of Bn}) into formula (\ref{General epsilon for non-simply laced systems}) gives the same value of $\widetilde{\gamma}_{(B_N,c)}.$ 
Similarly, we have 
$$\widetilde{\gamma}_{(C_N,c)}^2=-p\big(2q+(N-2)p\big),$$
which is equal to the value given by formula (\ref{General epsilon for non-simply laced systems}) after substitution $a_i$ from (\ref{ai of Cn}) and by using 
$\theta_{C_N} =2e^1=2(\alpha_1+\dots +\alpha_{N-1})+\alpha_{N}.$
It follows from Proposition \ref{F4 Martini's solution} that 
$$  \widetilde{\gamma}_{(F_4,c)}^2=-(p+q)(p+2q).$$
Note that
$\theta_{F_4}=e^1+e^2=2\alpha_1+3\alpha_2+4\alpha_3+2\alpha_4.$
Then it is easy to see that the substitution of values (\ref{ai of F4}) into formula (\ref{General epsilon for non-simply laced systems}) gives the same value of $\widetilde{\gamma}_{(F_4,c)}.$
 Similarly, it follows from Proposition \ref{G2 Martini's solution on the plane} that 
$$\widetilde{\gamma}_{(G_2,c)}^2=-\frac{3}{8}(p+q)(p+3q),$$
which is equal to the expression in formula (\ref{General epsilon for non-simply laced systems}) after the substitution of (\ref{ai of G2}) and by using
$\theta_{G_2}=\sqrt{3}e^1=2\alpha_1 +3\alpha_2.$
\end{proof}
It is not clear to us how to formulate Proposition \ref{Prop. values of epsilon} for any non-simply laced (reduced) root system in a uniform way.

Let us also give another formula for $\widetilde{\gamma}_{(\mathcal{R},c)}$ 
in terms of the dual root system ${\mathcal{R}^{\vee}=\{ \beta^{\vee}\colon\beta\in \mathcal{R}} \}$, 
where $ \beta^{\vee}=\frac{2\beta}{\langle \beta, \beta \rangle}.$
Then we have 
\begin{equation}\label{epsilon in terms of dual vectors}
    \widetilde{\gamma}_{(\mathcal{R},c)}^2=-\frac{{\langle \theta,\theta \rangle}^2}{32}\Big(a_0 \langle \theta^{\vee},\theta^{\vee} \rangle + \sum_{i=1}^{N}\overline{n}_{i}^{2} a_i\langle \alpha_{i}^{\vee},\alpha_{i}^{\vee} \rangle \Big),
\end{equation}
where coefficients $\overline{n}_{i}\in \mathbb{Z}_{\geq 0}$ are determined by the expansion ${\theta}^{\vee}=\sum_{i=1}^{N}\overline{n}_{i}\alpha_{i}^{\vee}.$
Formula (\ref{epsilon in terms of dual vectors}) follows from formula (\ref{General epsilon for non-simply laced systems}) by observing the relation 
$\overline{n}_{i}=\frac{n_{i} \langle \alpha_i , \alpha_i \rangle}{\langle  \theta, \theta \rangle}$ for $1\leq i\leq N .$

\section*{Acknowledgements}
M.F is grateful to L. Hoevenaars for collaboration in the beginning of the work, and to N. Nabijou and M. Pavlov for stimulating discussions. We are grateful to I. Strachan for pointing out paper \cite{Shen 2018} to us, and to B. Vlaar for advice on his Mathematica programme. The work of M.A was funded by Imam Abdulrahman Bin Faisal University, Kingdom of Saudi Arabia, Dammam.

%

\end{document}